\documentclass[a4paper,11pt]{amsart}
\usepackage{tikz}
\usepackage{amsmath,amssymb,url}
\usepackage[active]{srcltx}
\usepackage[T1]{fontenc}
\usepackage[latin1]{inputenc}
\usepackage{hyperref}
\usepackage{url}
\usepackage{latexsym,amssymb}
\usepackage{stmaryrd}
\usepackage{xcolor}
\usepackage{accents}
\usepackage{textcomp}
\usepackage{enumerate}
\usepackage[all]{xy}
\usepackage{mathtools}
\usepackage{scalerel,stackengine}
\stackMath
\usepackage[abbrev,non-sorted-cites]{amsrefs}
\BibSpec{article}{%
+{}{\PrintAuthors} {author}
+{,}{ \textit} {title}
+{,}{ } {journal}
+{}{ \textbf} {volume}
+{}{ \parenthesize} {date}
+{,}{ } {pages}
+{,}{ } {note}
+{.}{} {transition}
}
\BibSpec{collection.article}{%
    +{}  {\PrintAuthors}                {author}
    +{,} { \textit}                     {title}
    +{.} { }                            {part}
    +{:} { \textit}                     {subtitle}
    +{,} { \PrintContributions}         {contribution}
    +{,} { \PrintConference}            {conference}
    +{}  {\PrintBook}                   {book}
    +{, in} { }                         {booktitle}
    +{, Vol.} { }                       {volume}
    +{,} { \PrintDateB}                 {date}
    +{,} { pp.~}                        {pages}
    +{,} { }                            {status}
    +{,} { available at \eprint}        {eprint}
    +{}  { \PrintTranslation}           {translation}
    +{;} { \PrintReprint}               {reprint}
    +{.} { }                            {note}
    +{.} {}                             {transition}
}

\newcommand\reallywidehat[1]{%
\savestack{\tmpbox}{\stretchto{%
  \scaleto{%
    \scalerel*[\widthof{\ensuremath{#1}}]{\kern.1pt\mathchar"0362\kern.1pt}%
    {\rule{0ex}{\textheight}}%
  }{\textheight}%
}{2.4ex}}%
\stackon[-6.9pt]{#1}{\tmpbox}%
}

\numberwithin{equation}{section}

\renewcommand{\ge}{\geqslant}
\renewcommand{\le}{\leqslant}
\newdir^{((}{{}*!/-5pt/\dir_{(}}
\newdir^{(((}{{}*!/-5pt/\dir^{(}}
\let\op=\llbracket
\let\cl=\rrbracket

\let\cal=\mathcal
\def\Cl#1{\ensuremath{\cal{#1}}}
\newlength{\dhatheight}

\usepackage[cal=boondoxo]{mathalfa}

\newcommand\Stsynt{\mathrm{St}}
\newcommand\synt{\mathrm{Synt}}

\newcommand\vst[1]{\mathsf{#1}}
\newcommand\osc{\mathbf{OSC}_{\Omega}}
\newcommand\Rec{\operatorname{Rec}}

\theoremstyle{plain}
\newtheorem{Thm}{Theorem}[section]
\newtheorem{Prop}[Thm]{Proposition}
\newtheorem{Lemma}[Thm]{Lemma}

\newtheorem{Cor}[Thm]{Corollary}

\theoremstyle{definition}

\newtheorem{Remark}[Thm]{Remark}

\begin{document}

\title{Eilenberg correspondence for Stone recognition}
\thanks{The first author acknowledges partial support by CMUP (Centro de
Matem\'atica da Universidade do Porto), member of LASI (Intelligent
Systems Associate Laboratory), which is financed by Portuguese funds
through FCT (Funda\c c\~ao para a Ci\^encia e a Tecnologia, I. P.)
under the projects UIDB/00144/2020 and UIDP/00144/2020. 
}

\author[J. Almeida]{Jorge Almeida}%
\address{CMUP, Dep.\ Matem\'atica, Faculdade de Ci\^encias,
  Universidade do Porto, Rua do Campo Alegre 687, 4169-007 Porto,
  Portugal}
\email{jalmeida@fc.up.pt}

\author[O. Kl\'ima]{Ond\v rej Kl\'ima}%
\address{Dept.\ of Mathematics and Statistics, Masaryk University,
  Kotl\'a\v rsk\'a 2, 611 37 Brno, Czech Republic}%
\email{klima@math.muni.cz}


\makeatletter
\@namedef{subjclassname@2020}{%
  \textup{2020} Mathematics Subject Classification}
\makeatother

\begin{abstract}
  We develop and explore the idea of recognition of languages (in the
  general sense of subsets of topological algebras) as preimages of
  clopen sets under continuous homomorphisms into Stone topological
  algebras. We obtain an Eilenberg correspondence between varieties of
  languages and varieties of ordered Stone topological algebras and a
  Birkhoff/Reiterman-type theorem showing that the latter may me
  defined by certain pseudo-inequalities. In the case of classical
  formal languages, of words over a finite alphabet, we also show how
  this extended framework goes beyond the class of regular languages
  by working with Stone completions of minimal automata, viewed as
  unary algebras. This leads to a general method for showing that a
  language does not belong to a variety of languages, expressed in
  terms of sequences of pairs of words, which is illustrated when the
  class consists of all finite intersections of context-free
  languages.
\end{abstract}

\maketitle


\section{Introduction}
\label{sec:intro}

With the advent of electronic computers and their ever widening range
of applications, the theory of formal languages proved to be an
essential tool, not only in terms of the mathematical foundations of
computer science but also for the construction of efficient algorithms
and understanding the limitations of automatic computation. Various
means of describing languages have been explored, from general Turing
machines and many special cases such as finite automata, push-down
automata, linear bounded automata, to grammars and logic
specification. There is a very large body of literature in this area
with many important results exploring both the mathematical theory of
the models and their applications. The unfamiliar but interested
reader will get a feeling for this by having a look at
\cite{Rozenberg&Salomaa:1997;handb1,Pin:2021;HAT-I,Pin:2021;HAT-II}.

One of the aspects of this theory that has been particularly
successful explores the connections between regular languages and
finite semigroups via their minimal automata and syntactic semigroups.
A clear mathematical framework for the range of applications of this
method was put forth by Eilenberg \cite{Eilenberg:1976} in what came
to be known as \emph{Eilenberg's correspondence}, between
\emph{varieties of (regular) languages} and \emph{pseudovarieties of
  (finite) semigroups}. This approach was complemented by suitable
descriptions for such pseudovarieties given first in terms of
sequences of identities \cite{Eilenberg&Schutzenberger:1976} and later
in terms of pseudoidentities \cite{Reiterman:1982}. Always motivated
by problems in computer science, various extensions of this theory
have since been obtained to deal with generalizations of the notion of
variety of languages (which is defined by certain closure properties)
retaining an Eilenberg-type correspondence at the expense of modifying
the algebraic structure considered
\cite{Pin:1995d,Straubing:2002a}.

Only recently, there have been serious attempts at going beyond the
class of regular languages using the Eilenberg approach. One idea in
this direction is that languages over a fixed finite alphabet $A$ in a
variety form a Boolean algebra whose dual Stone space is a natural
profinite semigroup \cite{Almeida:1994a,Gehrke:2016a}. This suggests
taking the Boolean algebras of languages over a fixed finite alphabet
and exploring suitable algebraic structure in their dual space; two
alternative approaches to link this algebraic structure with
properties of the Boolean algebra of languages have been explored
\cite{Gehrke:2016a,Almeida&Klima:2024a}. We know however that this
method cannot take us beyond regular languages if we only consider
semigroups as all topological semigroups on a Stone space are
profinite \cite{Numakura:1957} and, therefore, can only be used to
recognize regular languages. Further proposals consist in relaxing the
definition of topological semigroup, such as requiring multiplication
to be continuous only in one factor
\cite{Steinberg:2013,Gehrke&Krebs:2017}.

Our approach consists in first getting a simple characterization of
the Boolean algebras of subsets of a topological algebra whose dual
spaces carry a natural topological algebra structure of the same kind
and which we view as the natural rich mathematical structure into
which the Boolean algebra is turned; this can be expressed in terms of
a transparent equivalence of categories. Such preliminary work has
been done in~\cite{Almeida&Klima:2024a}. In the present paper, after
suitable preparation done in Sections~\ref{sec:admissible}
and~\ref{sec:ordering}, we first show how Boolean algebras can be
relaxed to lattices with their dual (ordered) Priestley spaces
(Section~\ref{sec:syntactic}) and establish a general Eilenberg-type
correspondence between certain varieties of languages and certain
varieties of ordered Stone topological algebras
which is derived from a simple Galois connection
(Section~\ref{sec:Eilenberg}). We then show that our varieties of
ordered Stone topological algebras are defined by
\emph{pseudo-inequalities} (Section~\ref{sec:Reiterman}).

Unlike the case of regular languages, in which the Eilenberg/Reiterman
approach has been used mostly to obtain decidability results, we
envisage this generalization mostly to separate classes of languages
or simply to show that a particular language of interest cannot be
dealt with some type of computational means. The most famous
separation problem is the open question whether $\mathbf{P} =
\mathbf{NP}$ \cite{Fortnow:2019} but there are many other relevant
questions comparing various classes of languages defined by the
complexity of their membership problem.

In Section~\ref{sec:min-automata}, we show how our general theory
specializes to the case of unary algebras, where the operations
correspond to the letters of a fixed alphabet $A$. This leads to the
Stone/Priestley completion of the minimum automaton of an arbitrary
language over $A$, which is precisely the minimum compact automaton
considered by Steinberg \cite{Steinberg:2013}. This automaton has a
natural associated \emph{transition monoid} in the left
semi-topological Stone-\v Cech compactification $\beta(A^*)$ of the
discrete free monoid $A^*$, although we have not explored this aspect
in the present paper. The fact that our general varieties of languages
can be defined by pseudo-inequalities translates into a simple
separation criterion in terms of sequences of words which is akin to
the Eilenberg-Sch\"utzenberger ultimate characterization of
pseudovarieties by sequences of identities
\cite{Eilenberg&Schutzenberger:1976}.

Finally, in Section~\ref{sec:Pumping-CF}, we give examples of
application of how the general separation method developed in the
previous section can be used. The examples show that certain languages
are not intersections of finitely many context-free languages. This
class of languages has received some attention in the literature
\cite{Liu&Weiner:1973,Cislaru:1974,Paun:1979,Gorun:1980,Latta:1993}.
As we do, both \cite{Liu&Weiner:1973,Gorun:1980} use Parikh vectors to
infer properties of languages from this class.

\section{M-closed admissible Boolean algebras}
\label{sec:admissible}

In this section, we present the background and complement the results
from~\cite{Almeida&Klima:2024a} which are needed for our present
purposes.

Let $\mathbb{N}$ be set of all non-negative integers. By a
\emph{signature}, we mean an $\mathbb{N}$-graded set
$\Omega=\biguplus_{n\in\mathbb{N}}\Omega_n$. An
\emph{$\Omega$-algebra} is 
a nonempty set $S$ endowed with an
\emph{evaluation mapping} $E_n:\Omega_n\times S^n\to S$; the
alternative notation $o_S(s_1,\ldots,s_n)=E_n(o,s_1,\ldots,s_n)$ is
adopted in this paper. The notions of homomorphism, subalgebra, direct
product are standard, see \cite{Burris&Sankappanavar:1981}.

We follow \cite{Schneider&Zumbragel:2017} for the notion of
\emph{topological algebra}. A \emph{topological signature} is a
signature $\Omega$ where  
each $\Omega_n$ has a topological
structure. A \emph{topological $\Omega$-algebra} is an
$\Omega$-algebra $S$ endowed with a topology such that each evaluation
mapping 
$E_n$ is continuous. If, moreover, $S$ is a Stone space, that is, a
compact zero-dimensional space, then $S$ is said to be a \emph{Stone
  topological algebra}. 

For the remainder of this section, $\Omega$ is an arbitrary topological
signature and $T$ is an arbitrary topological $\Omega$-algebra.

By a \emph{Stone completion} of $T$ we mean a continuous homomorphism
$\varphi:T\to S$ into a Stone topological algebra $S$ whose image is
dense in~$S$. We gave in \cite[Theorem~6.11]{Almeida&Klima:2024a} a
characterization of the Boolean subalgebras \Cl B of the Boolean
algebra~$\Cl P_{co}(T)$ of all clopen subsets of~$T$ for which the
Stone dual space $\Cl B^\star$ has a structure of topological algebra
such that the natural mapping $\iota_{\Cl B}:T\to\Cl B^\star$ (sending
each $t\in T$ to the ultrafilter consisting of all $L\in\Cl B$ such
that $t\in L$) is a Stone completion of~$T$. Such Boolean algebras are
called \emph{$\Omega$-finite (in T)}. Recall from Stone duality theory
that a basis of the topology of~$\Cl B^\star$ is given by the clopen
sets
\begin{displaymath}
  \Cl U^{\Cl B}(L)
  =\Cl U_L^{\Cl B}
  =\{u\in\Cl B^\star:L\in u\}
\end{displaymath}
with $L\in\Cl B$, which are such that $\iota_{\Cl B}^{-1}(\Cl U_L^{\Cl
  B})=L$ and $\overline{\iota_{\Cl B}(L)}=\Cl U_L^{\Cl B}$.

We also showed in \cite[Lemma~6.4]{Almeida&Klima:2024a} that there is a
maximum $\Omega$-finite Boolean subalgebra of~$\Cl P_{co}(T)$, which
is denoted $\Cl B_{\mathrm{max}}^T$. The members of $\Cl
B_{\mathrm{max}}^T$ are said to be \emph{admissible subsets} of~$T$;
sets of admissible subsets of~$T$ may also be said to be
\emph{admissible (in~$T$)}. By
\cite[Proposition~6.3]{Almeida&Klima:2024a}, a Boolean algebra $\Cl B$
of admissible subsets of~$T$ is $\Omega$-finite if and only if it
satisfies the following property:
\begin{equation}
  \label{eq:w-Om-fin}
  \forall n\ge0\
  \forall o\in\Omega_n\
  \left(L\in\Cl B
  \implies
  (o_T)^{-1}(L)\in\Cl B^{(n)}\right)
\end{equation}
where the tensor product $\Cl B^{(n)}$ of $n$ copies of~$\Cl B$ is
viewed as the Boolean subalgebra of~$\Cl P_{co}(T^n)$ consisting of
the finite unions of \emph{boxes} of the form $L_1\times\cdots\times
L_n$ with the $L_i$ in~$\Cl B$.

Recall that, for a topological algebra $T$, $M(T)$ denotes the monoid
of \emph{linear transformations} of~$T$, that is, self-mappings of $T$
of the form
\begin{displaymath}
  f(t)=w_T(t,t_2,\ldots,t_n) \quad(t\in T)
\end{displaymath}
where $t_2,\ldots,t_n$ are fixed elements of~$T$ and
$w(x_1,\ldots,x_n)$ is an $\Omega$-term in the variables
$x_1,\ldots,x_n$ in which $x_1$ occurs only once. We say that a subset
$\Cl C$ of~$\Cl P(T)$ is \emph{M-closed} if, for every $L\in\Cl C$ and
$f\in M(T)$, we also have $f^{-1}(L)\in\Cl C$.

\begin{Prop}
  \label{p:M-closed}
  Let \Cl B be a Boolean subalgebra of~$\Cl P_{co}(T)$. If \Cl B
  satisfies Condition~\eqref{eq:w-Om-fin} then $\Cl B$ is~M-closed.
  Conversely, if \Cl B is M-closed and admissible, then \Cl B
  satisfies Condition~\eqref{eq:w-Om-fin}.
\end{Prop}

\begin{proof}
  Suppose first that $\Cl B$ satisfies Condition~\eqref{eq:w-Om-fin}.
  To show that \Cl B is M-closed, it suffices to consider linear terms
  of height 1, that is, we may assume that
  $f(x)=o_T(t_1,\ldots,t_{i-1},x,t_{i+1},\ldots,t_n)$ for some
  elements $t_\ell\in T$ and operation symbol $o\in\Omega_n$ and show
  that, for every $L\in\Cl B$, we have $f^{-1}(L)\in\Cl B$. By
  assumption, since $L\in\Cl B$, there is a finite decomposition
  \begin{equation}
    \label{eq:M-closed-1}
    (o_T)^{-1}(L)=\bigcup_{j=1}^r L_{j,1}\times\cdots\times L_{j,n}
  \end{equation}
  where $L_{j,\ell}\in\Cl B$ ($j=1,\ldots,r$; $\ell=1,\ldots,n$).
  Consider the set
  \begin{displaymath}
    \tilde L=\bigcup\left\{L_{j,i}:\ t_\ell\in L_{j,\ell}\ 
      (\ell\in\{1,\ldots,n\}\setminus \{i\})\right\}.
  \end{displaymath}
  To complete the proof that $f^{-1}(L)\in\Cl B$, it remains to
  observe that $f^{-1}(L)=\tilde L$ as $\tilde L$ is a finite union of
  members of~\Cl B.

  For the converse, suppose that
  $o\in\Omega_n$ and
  $L\in\Cl B$. Since we are assuming
  that $L$ is admissible, there is a finite decomposition
  \eqref{eq:M-closed-1} with all $L_{j,\ell}\in\Cl
  B_{\mathrm{max}}^T$. By \cite[Lemma~4.1]{Almeida&Klima:2024a}, we may
  assume that, for each $\ell$ the factors $L_{j,\ell}$ are members of
  a finite partition of $T$ into members of~$\Cl B_{\mathrm{max}}^T$
  such that, for two distinct such classes $K, K'$, there exist
  $t_i\in T$ ($i\in\{1,\ldots,n\}\setminus\{\ell\}$) such that exactly
  one of the sets
  \begin{align*}
    &o_T(t_1,\ldots,t_{\ell-1},K,t_{\ell+1},\ldots,t_n)\\
    &o_T(t_1,\ldots,t_{\ell-1},K',t_{\ell+1},\ldots,t_n)
  \end{align*}
  is contained in~$L$. In this way, considering all such pairs of
  classes of the partition, we obtain a finite set $F$ of linear
  transformations of~$T$ such that, for each class $K$ of the
  partition,
  \begin{displaymath}
    K=
    \bigcap_{f\in F\atop K\subseteq f^{-1}(L)}f^{-1}(L)
    \setminus 
    \bigcup_{f\in F\atop K\nsubseteq f^{-1}(L)}f^{-1}(L).
  \end{displaymath}
  In particular, we conclude that each $L_{j,\ell}$ belongs to $\Cl B$
  and so $\Cl B$ satisfies Condition~(\ref{eq:w-Om-fin}).
\end{proof}

For a set \Cl C of admissible subsets of~$T$, we denote by $\langle\Cl
C\rangle$ the smallest $\Omega$-finite Boolean subalgebra of~$\Cl
P_{co}(T)$ containing~$\Cl C$, whose existence is guaranteed by
\cite[Corollary~6.8]{Almeida&Klima:2024a}. For an admissible subset
$L$ of~$T$, we write $\langle L\rangle$ instead of
$\langle\{L\}\rangle$. We may also emphasize that the $\Omega$-finite
Boolean algebra $\langle L\rangle$ consists of admissible subsets
of~$T$ by writing~$\langle L\rangle_T$.

\begin{Cor}
  \label{c:B(L)}
  Let $L$ be an admissible subset of\/~$T$. Then $\langle L\rangle$ is
  the Boolean subalgebra of~$\Cl P_{co}(T)$ generated by all sets of the
  form $f^{-1}(L)$, where $f\in M(T)$.
\end{Cor}

\begin{proof}
  Let $\Cl B$ denote the Boolean algebra described in the statement of
  the corollary. Since $\Cl B$ is M-closed, it is precisely the
  smallest Boolean subalgebra of $\Cl B_{\mathrm{max}}^{T}$ containing
  $L$ that is M-closed. The result now follows from
  Proposition~\ref{p:M-closed}.
\end{proof}

The following proposition and its corollary relate $\Omega$-finiteness
in a topological algebra and in its maximum Stone completion.

\begin{Prop}
  \label{p:Om-fin-vs-U}
  Let $T$ be a topological algebra and $\Cl B=\Cl B_{\mathrm{max}}^T$.
  Then a Boolean subalgebra $\Cl C$ of~$\Cl B$ is $\Omega$-finite in
  $T$ if and only if $\Cl U^{\Cl B}(\Cl C)$ is $\Omega$-finite in $\Cl
  B^\star$.
\end{Prop}

\begin{proof}
  Suppose first that $\Cl C$ is $\Omega$-finite in $T$. By
  \cite[Corollary~7.4]{Almeida&Klima:2024a}, there is a continuous
  homomorphism $\delta:\Cl B^\star\to \Cl C^\star$ such that
  $\delta\circ\iota_{\Cl B}=\iota_{\Cl C}$. Since $\iota_{\Cl B}^{-1}$
  is the inverse of~$\Cl U^{\Cl B}$, we obtain the following
  commutative diagram, where $\Cl B_\delta=\delta^{-1}\bigl(\Cl
  P_{co}(\Cl C^\star)\bigr)$, which is $\Omega$-finite by
  \cite[Corollary~6.10]{Almeida&Klima:2024a}:
  \begin{displaymath}
    \xymatrix@C15mm{
      \Cl C
      \ar@{^{(((}->}[r]
      &
      \Cl B
      \ar@/^2mm/[r]^(.45){\Cl U^{\Cl B}}
      &
      {\Cl P_{co}(\Cl B^\star)}
      \ar@/^2mm/[l]^(.55){\iota_{\Cl B}^{-1}}
      \\
      \Cl P_{co}(\Cl C^\star)
      \ar[u]^{\iota_{\Cl C}^{-1}}
      \ar[rr]^{\delta^{-1}}
      &&
      \Cl B_\delta
      \ar@{^{(((}->}[u]
    }
  \end{displaymath}
  Hence $\Cl U^{\Cl B}(\Cl C)=\Cl B_\delta$ is $\Omega$-finite in~$\Cl
  B^\star$.

  Conversely, if $\Cl U^{\Cl B}(\Cl C)$ is $\Omega$-finite in~$\Cl
  B^\star$, then $\Cl C=\iota_{\Cl B}^{-1}\bigl(\Cl U^{\Cl B}(\Cl
  C)\bigr)$ is $\Omega$-finite in $T$ by
  \cite[Proposition~6.9]{Almeida&Klima:2024a}.
 \end{proof}

\begin{Cor}
  \label{c:Om-fin-gen-vs-U}
  Let $L$ be an admissible subset of a topological algebra $T$ and let
  $\Cl B=\Cl B_{\mathrm{max}}^T$. Then the equality $\langle \Cl
  U^{\Cl B}_L\rangle_{\Cl B^\star}=\Cl U^{\Cl B}(\langle L\rangle_T)$ holds.
\end{Cor}

\begin{proof}
  Let $K=\Cl U^{\Cl B}_L$. By Proposition~\ref{p:Om-fin-vs-U}, $(\Cl
  U^{\Cl B})^{-1}(\langle K\rangle_{\Cl B^\star})$ is an
  $\Omega$-finite Boolean algebra in~$T$; since $L=(\Cl U^{\Cl
    B})^{-1}(K)$ is one of its members, it must contain $\langle
  L\rangle_T$. Similarly, $\Cl U^{\Cl B}(\langle L\rangle_T)$ is an
  $\Omega$-finite Boolean 
  subalgebra in~$\Cl B^\star$ of which $K=\Cl U^{\Cl
    B}(L)$ is a member, whence $\Cl U^{\Cl B}(\langle L\rangle_T)$
  contains $\langle K\rangle_{\Cl B^\star}$. Combining the two
  inclusions, we obtain the desired equality.
\end{proof}

Let $\varphi:T\to S$ be a continuous homomorphism of topological
algebras with dense image. Given a linear transformation $f$ of $T$,
we may choose a term $w\in T_\Omega(\{x_1,\ldots,x_n\})$ linear in
$x_1$ and $t_2,\ldots,t_n\in T$ such that
\begin{equation}
  \label{eq:lin-transf-of_T}
  f(t)=w_T(t,t_2,\ldots,t_n)
  \quad(t\in T).
\end{equation}
Then, there is an associated linear transformation of~$S$ defined by
\begin{equation}
  \label{eq:induced-lin-transf}
  g(s)=w_S(s,\varphi(t_2),\ldots,\varphi(t_n))
  \quad(s\in S).
\end{equation}
Since $\varphi$ is a homomorphism, we get $\varphi\circ
f=g\circ\varphi$. Hence, the restriction of~$g$ to the image
of~$\varphi$ does not depend on the choice of $(w,t_2,\ldots,t_n)$. As
$g$ is continuous and the image of~$\varphi$ is dense in~$S$, we
conclude that $g$ itself is independent of that choice. We call $g$
the \emph{linear transformation of $S$ induced by~$f$} and we denote
by $M_\varphi(S)$ the set of all such $g$ as $f$ runs over $M(T)$.

\begin{Lemma}
  \label{l:induced-linear-in-Stsynt}
  Let $\Cl B$ be an $\Omega$-finite Boolean subalgebra of~$\Cl
  P_{co}(T)$ and $L$ be a member of~$\Cl B$. Given $u\in\Cl B^\star$
  and $f\in M(T)$ with induced $g\in M(\Cl B^\star)$, $f^{-1}(L)\in u$
  holds if and only if so does $u\in g^{-1}(\Cl U_L)$.
\end{Lemma}

\begin{proof}
  Let $(u_i)_{i\in I}$ be a net in $T$ such that $u=\lim\iota_{\Cl
    B}(u_i)$. Let $K=\Cl U_L$.

  Suppose first that $g(u)\in K$. Since the mapping $g$ is continuous,
  we have $\lim g(\iota_{\Cl B}(u_i))=g(u)$. As $K$ is open and
  $\iota_{\Cl B}$ is a homomorphism, we may assume that $\iota_{\Cl
    B}(f(u_i))=g(\iota_{\Cl B}(u_i))\in K$ for all $i\in I$.
  Equivalently, we get $f(u_i)\in \iota_{\Cl B}^{-1}(K)=L$ for all
  $i\in I$. Now, the condition $u_i\in f^{-1}(L)$ holds if and only if
  so does $f^{-1}(L)\in\iota_{\Cl B}(u_i)$, which means that
  $\iota_{\Cl B}(u_i)\in\Cl U_{f^{-1}(L)}$. Finally, $u$ belongs to
  the clopen set $\Cl U_{f^{-1}(L)}$ if and only if there exists
  $i_0\in I$ such that, for every $i\ge i_0$, $\iota_{\Cl
    B}(u_i)\in\Cl U_{f^{-1}(L)}$. It follows that $f^{-1}(L)\in u$, as
  required. As, taking into account that $K$ is closed, all the above
  steps are reversible, the converse also holds.
\end{proof}

The following lemma shows the relevance of the induced linear
transformations.

\begin{Lemma}
  \label{l:reduction-to-induced-linear-transformations}
  Let $\varphi:T\to S$ be a Stone completion and let $K$ be a clopen
  subset of~$S$. Then, for every $h\in M(S)$ there is $g\in
  M_\varphi(S)$ such that $h^{-1}(K)=g^{-1}(K)$.
\end{Lemma}

\begin{proof}
  Assume first that $h^{-1}(K)\ne\emptyset$. There exist $w\in
  T_\Omega(\{x_1,\ldots,x_n\})$ linear in $x_1$ and $s_2,\ldots,s_n\in
  S$ such that
  \begin{displaymath}
    h(s)=w_S(s,s_2,\ldots,s_n)
    \quad(s\in S).
  \end{displaymath}
  Since $S$ is a topological algebra, the mapping $w_S:S^n\to S$ is
  continuous. As $S$ is a Stone space, it follows that there are
  $m\ge0$, nonempty clopen subsets $K_i$ of~$S$ and $P_i$ of~$S^{n-1}$
  such that
  \begin{equation*}
    w_S^{-1}(K)=\bigcup_{i=1}^m K_i\times P_i,
  \end{equation*}
  where we may assume that the sets $P_i$ are pairwise disjoint. If
  $(s_2,\ldots,s_n)$ belongs to no $P_i$, then $h^{-1}(K)=\emptyset$,
  contrary to our initial assumption. Hence, there is a unique $i$
  such that $(s_2,\ldots,s_n)\in P_i$ and we may choose a point
  $(t_2,\ldots,t_n)$ in the set $(\varphi^{n-1})^{-1}(P_i)$. Note that
  \begin{displaymath}
    w_S(s,s_2,\ldots,s_n)\in K
    \iff
    s\in K_i
    \iff
    w_S(s,\varphi(t_2),\ldots,\varphi(t_n))\in K.
  \end{displaymath}
  Thus, $h^{-1}(K)=K_i=g^{-1}(K)$, where $g\in M_\varphi(S)$ is
  defined by $g(s)=w_S(s,\varphi(t_2),\ldots,\varphi(t_n))$.

  Finally, we consider the case where $h^{-1}(K)=\emptyset$. 
  Let $K^\complement$ denote the complement of $K$ in $S$.
  Then
  $h^{-1}(K^\complement)=S\ne\emptyset$ and so, by the first part of
  the proof, there is $g\in M_\varphi(S)$ such that
  $S=h^{-1}(K^\complement)=g^{-1}(K^\complement)$, whence
  $h^{-1}(K)=\emptyset=g^{-1}(K)$.
\end{proof}

Combining Lemmas~\ref{l:induced-linear-in-Stsynt}
and~\ref{l:reduction-to-induced-linear-transformations}, we obtain the
following result.

\begin{Prop}
  \label{p:M-closed-in-Stone-completion}
  Let $\varphi:T\to S$ be a Stone completion and $\Cl C$ be a subset
  of~$\Cl P_{co}(S)$. Then the following conditions are equivalent:
  \begin{enumerate}
  \item\label{item:M-closed-in-Stone-completion-1} the set $\Cl C$ is
    M-closed in~$S$;
  \item\label{item:M-closed-in-Stone-completion-2}
    for every $K\in\Cl C$ and every $g\in M_\varphi(S)$, $g^{-1}(K)$
    belongs to~$\Cl C$;
  \item\label{item:M-closed-in-Stone-completion-3} the set
    $\varphi^{-1}(\Cl C)$ is M-closed in~$T$;
  \item\label{item:M-closed-in-Stone-completion-4} there exists an
    M-closed subset $\Cl D$ of $\varphi^{-1}(\Cl P_{co}(S))$ such that
    $\Cl C=\{\overline{{\varphi(L)}}:L\in\Cl D\}$.
  \end{enumerate}
\end{Prop}

\begin{proof}
  The implication
  $(\ref{item:M-closed-in-Stone-completion-1})\Rightarrow
  (\ref{item:M-closed-in-Stone-completion-2})$ is obvious.

  For $(\ref{item:M-closed-in-Stone-completion-2})\Rightarrow
  (\ref{item:M-closed-in-Stone-completion-3})$, let
  $L=\varphi^{-1}(K)$ with $K\in\Cl C$. Given $f\in M(T)$, let $g\in
  M_\varphi(S)$ be the induced linear transformation of~$S$, so that
  $\varphi\circ f=g\circ\varphi$. Then, for each $t\in T$, the
  following equivalences hold:
  \begin{align*}
    t\in f^{-1}(L)
    &\iff f(t)\in\varphi^{-1}(K) \\
    &\iff \varphi(f(t))\in K\\
    &\iff g(\varphi(t))\in K\\
    &\iff t\in \varphi^{-1}(g^{-1}(K)).
  \end{align*}
  Hence, $f^{-1}(L)=\varphi^{-1}(g^{-1}(K))$ so that, in view
  of~(\ref{item:M-closed-in-Stone-completion-2}), $g^{-1}(K)$ belongs
  to~$\Cl C$ and $f^{-1}(L)$ 
  belongs to~$\varphi^{-1}(\Cl C)$.

  To prove $(\ref{item:M-closed-in-Stone-completion-3})\Rightarrow
  (\ref{item:M-closed-in-Stone-completion-4})$, it suffices to observe
  that, for every clopen subset $K$ of~$S$,
  $\overline{\varphi(\varphi^{-1}(K))}=K$ because $\varphi$ is a
  mapping with dense image while $\varphi^{-1}(K)$ is admissible by
  \cite[Corollary~6.10]{Almeida&Klima:2024a}.
  
  Finally, to establish
  $(\ref{item:M-closed-in-Stone-completion-4})\Rightarrow
  (\ref{item:M-closed-in-Stone-completion-1})$, suppose that $L$ is an
  admissible subset of~$T$ and let $K=\overline{{\varphi(L)}}$ and
  $h\in M(S)$. Note that $K$ is clopen in~$S$ and $\varphi^{-1}(K)=L$.
  By Lemma~\ref{l:reduction-to-induced-linear-transformations}, there
  is $f\in M(T)$ such that the induced linear transformation $g$
  of~$S$ satisfies $h^{-1}(K)=g^{-1}(K)$. We claim that
  \begin{equation}
    \label{eq:M-closed-in-Stone-completion}
    g^{-1}(K)=\overline{\varphi\bigl(f^{-1}(L)\bigr)}.
  \end{equation}
  Since the image of $\varphi$ is dense in~$S$ and both sides
  of~\eqref{eq:M-closed-in-Stone-completion} are clopen sets, the
  equality \eqref{eq:M-closed-in-Stone-completion} holds if and only
  if both sides contain the same elements of the form $\varphi(t)$
  with $t\in T$, a property that is a consequence of the following
  chain of equivalences:
  \begin{align*}
    \varphi(t)\in g^{-1}(K)
    &\iff g(\varphi(t))\in K \\
    &\iff \varphi(f(t))\in K\\
    &\iff f(t)\in\varphi^{-1}(K)=L\\
    &\iff \varphi(t)\in \varphi(f^{-1}(L)).
  \end{align*}
  This establishes the claim, which yields the desired implication.
\end{proof}

\section{Ordering Stone topological algebras}
\label{sec:ordering}

A binary relation $\rho$ on an $\Omega$-algebra $S$ is said to be
\emph{stable} if, for all $n\ge1$, $o\in\Omega_n$, and $s_i,t_i\in S$,
\begin{equation}
  \label{eq:ordered-Stone}
  s_i\mathrel{\rho} t_i\ (i=1,\ldots,n)
  \implies
  o_S(s_1,\ldots,s_n) \mathrel{\rho} o_S(t_1,\ldots,t_n).
\end{equation}
We say that a pair $(S,{\preccurlyeq})$ is a \emph{quasi-ordered
  topological algebra} if $\preccurlyeq$ is a stable closed
quasi-order on~$S$.
An \emph{ordered Stone topological algebra} is a
quasi-ordered topological algebra $(S,{\le})$ in which $\le$ is
a partial order.

If $\preccurlyeq$ is a quasi-order on a set $S$ then a subset $L$
of~$S$ is said to be an \emph{$\preccurlyeq$-upset} if $s\in L$ and
$s\preccurlyeq t$ implies $t\in L$.
In case $(S,{\preccurlyeq})$ is a
quasi-ordered topological algebra, we denote by $\Cl P_{uco}(S)$ the
set of all clopen $\preccurlyeq$-upsets of~$S$.

\begin{Lemma}
  \label{l:Puco}
  For a quasi-ordered topological algebra $(S,{\preccurlyeq})$, the
  set $\Cl P_{uco}(S)$ is an $M$-closed $(0,1)$-sublattice of~$\Cl
  P_{co}(S)$.
\end{Lemma}

\begin{proof}
  Certainly the empty set and $S$ belong to~$\Cl P_{uco}(S)$ and it is
  easy to check that $\Cl P_{uco}(S)$ is closed under binary union and
  binary intersection. Suppose now that $K\in\Cl P_{uco}(S)$ and let
  $w\in T_\Omega(\{x_1,\ldots,x_n\})$, $s_2,\ldots,s_n\in S$ and
  $f(s)=w_S(s,s_2,\ldots,s_n)$ for each $s\in S$. If
  $s\preccurlyeq s'$, then $f(s)\preccurlyeq f(s')$ and so
  \begin{displaymath}
    s\in f^{-1}(K)
    \iff
    f(s)\in K
    \implies
    f(s')\in K
    \iff
    s'\in f^{-1}(K)
  \end{displaymath}
  which shows that $f^{-1}(K)\in\Cl P_{uco}(S)$ since $f$ is
  continuous.
\end{proof}

Say that a subset $\Cl L$ of $\Cl P(T)$ is \emph{P-closed} if,
whenever $f$ is a polynomial transformation of~$T$, $L\in\Cl L$
implies $f^{-1}(L)\in\Cl L$. The above proof shows that the set of
clopen upsets of a quasi-ordered topological algebra is actually
P-closed. Yet, it is not true that M-closed subsets of a topological
algebra are always P-closed. For instance, the set of all context-free
languages over a three-letter alphabet is M-closed but not P-closed
\cite{MSE1573677:2015}. We do not know whether in case $\Cl L$ is an
M-closed lattice of subsets of~$T$, then $\Cl L$ is also P-closed.

For a quasi-order $\preccurlyeq$ on a set $S$, we may consider the
associated equivalence relation $\equiv$ given by
${\preccurlyeq}\cap{\succcurlyeq}$. On the quotient set $S/{\equiv}$
we get an induced partial order $\le$ such that $(s/{\equiv})\le
(t/{\equiv})$ if and only if $s\preccurlyeq t$. In case $S$ is an
$\Omega$-algebra and $\preccurlyeq$ is a stable quasi-order on~$S$,
the equivalence relation $\equiv$ is a congruence on $S$ and so
$S/{\equiv}$ inherits a structure of~$\Omega$-algebra. On the other
hand, if $S$ is a topological space, we may consider on $S/{\equiv}$
the quotient topology. It is well known that, in case $S$ is compact
and $\preccurlyeq$ is closed in~$S^2$, the quotient space $S/{\equiv}$
is also compact. But, if $S$ is a Stone space and $\preccurlyeq$ is a
closed quasi-order on~$S$, the quotient space $S/{\equiv}$ may not be
zero-dimensional. A property assuring that $S/{\equiv}$ is
zero-dimensional is the following \emph{Priestley condition}
\begin{equation}
  \label{eq:Priestley-property}
  \forall s,t\in S\
  \bigl( s\not\preccurlyeq t
  \implies \exists L\in\Cl P_{uco}(S)\
  (s\in L\wedge t\notin L) \bigr).
\end{equation}
In fact, since $\preccurlyeq$-upsets are $\equiv$-saturated, we see
that the quotient space $(S/{\equiv},{\le})$ also
satisfies~(\ref{eq:Priestley-property}) for the induced partial
order~$\le$. Conversely, if $(S/{\equiv},{\le})$
satisfies~(\ref{eq:Priestley-property}) then so does
$(S,{\preccurlyeq})$. Note also that, if $(\preccurlyeq_i)_{i\in I}$
is a nonempty family of quasi-orders on the topological space $S$
satisfying the Priestley condition then the
intersection~$\preccurlyeq$ of the family is also a quasi-order
satisfying the Priestley condition, simply because an
$\preccurlyeq_i$-upset is also an $\preccurlyeq$-upset.

We say that the ordered Stone space $(S,{\le})$ is a \emph{Priestley
  space} if it satisfies the Priestley condition. In case $S$ is also
a topological $\Omega$-algebra for the same topology, then we say that
$(S,{\le})$ is a \emph{Priestley $\Omega$-algebra}. 

A further role of the Priestley condition is given by the following
result which is observed in~\cite{Priestley:1970}.

\begin{Lemma}
  \label{l:Priestley}
  If $(S,{\le})$ is a Priestley space then the Boolean algebra $\Cl
  P_{co}(S)$ is generated by the lattice $\Cl P_{uco}(S)$.
\end{Lemma}

The ``Ersatzkette'' of Stralka \cite[Proposition~(a)]{Stralka:1980} is
a counterexample to the converse of Lemma~\ref{l:Priestley} when the
signature $\Omega$ is empty. For a given sublattice of~$\Cl P_{co}(S)$
generating it as a Boolean algebra, there is a unique closed partial
order on $S$ satisfying the Priestley condition such that the members
of the lattice are the clopen upsets. This result is
part of the
Priestley duality theory and is generalized below when the Stone space
is a Stone completion of a topological algebra.

It is convenient to denote the complement of a subset $K$ in an
ambient set which should be clear from the context by~$K^\complement$.
The notation is somewhat abusive when we don't specify where the
complement is being taken. Yet, we trust that the reader will have no
problem figuring out from the context in which set a subset is being
complemented.

We say that $\varphi:T\to(S,{\le})$ is a \emph{Priestley stamp (over
  $T$)} if $\varphi:T\to S$ is a Stone completion and $(S,{\le})$ is a
Priestley space. Our aim is to show that, up to isomorphism, Priestley
stamps over $T$ are in bijection with the M-closed sublattices of $\Cl
B_{\mathrm{max}}^T$. Equivalently, such lattices are in bijection with
the closed stable quasi-orders on the Stone topological algebra $(\Cl
B_{\mathrm{max}}^T)^\star$ that satisfy the Priestley condition. Our
next proposition is a preparatory result in that direction. Before
stating it note that, given a subset $\Cl C$ of~$\Cl P(T)$ and a
Boolean subalgebra $\Cl B$ of~$\Cl P(T)$ containing~$\Cl C$, consider
the following binary relation on~$\Cl B^\star$: for ultrafilters
$u,v\in\Cl B^\star$, let
\begin{equation}
  \label{eq:qo-for-M-closed-admissible}
  u\preccurlyeq_{\Cl C} v
  \quad\text{if\/ }
  \forall L\in\Cl C\ (L\in u \implies L\in v).
\end{equation}
The relation $\preccurlyeq_{\Cl C}$ is clearly a quasi-order on~$\Cl
B^\star$. To show that $\preccurlyeq_{\Cl C}$ is closed, note that it
is the intersection of the relations $S^2\setminus(\Cl U_L\times\Cl
U_{L^\complement})$ ($L\in\Cl C$), which are closed.

\begin{Prop}
  \label{p:M-closed-admissible-to-qo}
  Let $\Cl C$ be an M-closed subset of $\Cl B_{\mathrm{max}}^T$. Then
  the pair $\bigl((\Cl B_{\mathrm{max}}^T)^\star,{\preccurlyeq_{\Cl
      C}}\bigr)$ is a quasi-ordered Stone topological algebra that
  satisfies the Priestley condition.
\end{Prop}

\begin{proof}
  To simplify the notation used in this proof, let $S=(\Cl
  B_{\mathrm{max}}^T)^\star$, $\varphi=\iota_{\Cl B_{\mathrm{max}}^T}$
  and denote $\preccurlyeq_{\Cl C}$ by $\preccurlyeq$. By the above,
  we already know that $\preccurlyeq$ is a closed quasi-order on $S$.

  To show that $\preccurlyeq$ is stable, first note that it suffices
  to show that it is preserved by linear transformations of~$S$. So,
  suppose that $u\preccurlyeq v$ and let $h\in M(S)$. Let $L\in\Cl C$
  and assume that $L\in h(u)$, that is, $h(u)\in \Cl U_L$, which
  further means that $u\in h^{-1}(\Cl U_L)$. By
  Lemma~\ref{l:reduction-to-induced-linear-transformations}, there is
  $f\in M(T)$ such that, for the induced $g\in M_\varphi(S)$,
  $h^{-1}(\Cl U_L)=g^{-1}(\Cl U_L)$. By
  Lemma~\ref{l:induced-linear-in-Stsynt}, we know that $u\in
  g^{-1}(\Cl U_L)$ if and only if $f^{-1}(L)\in u$. Since $\Cl C$ is
  M-closed and $u\preccurlyeq v$, we get $f^{-1}(L)\in v$. Reversing
  the above steps with $v$ in the place of~$u$, we conclude that $L\in
  h(v)$, thereby showing that $h(u)\preccurlyeq h(v)$, which proves
  that $\preccurlyeq$ is stable.
  
  It remains to show that $\preccurlyeq$ satisfies the Priestley
  condition. Indeed if $u,v\in S$ and $u\not\preccurlyeq v$, then
  there exists $L\in\Cl C$ such that $L\in u$ and $L\notin v$, that is
  the clopen set $\Cl U_L$ contains $u$ but not $v$. To conclude, it
  suffices to observe that $\Cl U_L$ is an $\preccurlyeq$-upset.
\end{proof}

\section{Syntactic Priestley algebra}
\label{sec:syntactic}

As in the preceding section, we assume throughout this section that
$\Omega$ is an arbitrary topological signature and $T$ is an arbitrary
topological $\Omega$-algebra.

Let $\Cl C$ be a subset of~$\Cl B_{\mathrm{max}}^T$ and let
$\varphi:(\Cl B_{\mathrm{max}}^T)^\star\to \langle\Cl C\rangle^\star$
be the associated continuous homomorphism, which is defined by
$\varphi(u)=u\cap\langle\Cl C\rangle$. Consider the least M-closed set
containing $\Cl C$:
\begin{equation}
  \label{eq:M-closure}
  \Cl C_M=\{f^{-1}(L):f\in M(T), \ L\in\Cl C\}.
\end{equation}
Let $[\Cl C]$ the $(0,1)$-sublattice of~$\Cl B_{\mathrm{max}}^T$
generated by $\Cl C_M$, which is the smallest M-closed
$(0,1)$-sublattice of~$\Cl B_{\mathrm{max}}^T$ containing $\Cl C$. Let
$\Stsynt(\Cl C)$ be the Stone topological algebra $\langle\Cl
C\rangle^\star$ and denote by $\eta_{\Cl C}$ the Stone completion
$\iota_{\langle\Cl C\rangle}$. We endow $\Stsynt(\Cl C)$ with the
binary relation defined by
\begin{equation}
  \label{eq:po-for-admissible}
  u\le_{\Cl C} v
  \quad\text{if\/ }
  \forall L\in\Cl C_M\ (L\in u \implies L\in v).
\end{equation}
Then, the equality $\varphi(u)=\varphi(v)$ holds if and only if the
inequalities $u\preccurlyeq_{\Cl C_M}v$ and $v\preccurlyeq_{\Cl C_M}
u$ are both valid, where the quasi-order $\preccurlyeq_{\Cl C_M}$ is
that defined in~(\ref{eq:qo-for-M-closed-admissible}). Moreover,
$u\preccurlyeq_{\Cl C_M}v$ is equivalent to $\varphi(u)\le_{\Cl
  C}\varphi(v)$. This proves the following corollary of
Proposition~\ref{p:M-closed-admissible-to-qo}.

\begin{Cor}
  \label{c:admissible-to-po}
  Let $\Cl C$ be a subset of $\Cl B_{\mathrm{max}}^T$ and $\equiv_{\Cl
    C_M}$ be the congruence ${\preccurlyeq_{\Cl
      C_M}}\cap{\succcurlyeq_{\Cl C_M}}$. Then $\bigl(\Stsynt(\Cl
  C),{\le_{\Cl C}}\bigr)$ is a Priestley algebra isomorphic to the
  quotient $\bigl((\Cl B_{\mathrm{max}}^T)^\star/{\equiv_{\Cl
      C_M}},{\preccurlyeq_{\Cl C_M}}\bigr)$ and, therefore, $\eta_{\Cl
    C}:T\to(\Stsynt(\Cl C),{\le_{\Cl C}})$ is a Priestley stamp.
\end{Cor}

Further properties of the quasi-orders $\preccurlyeq_{\Cl C_M}$ are
presented in the following result.

\begin{Prop}
  \label{p:admissible-sets-vs-quasi-orders}
  Let $\Cl C$ and $\Cl D$ be subsets of~$\Cl B_{\mathrm{max}}^T$.
  \begin{enumerate}
  \item\label{item:admissible-sets-vs-quasi-orders-1} If $\Cl
    C\subseteq\Cl D$ then $\preccurlyeq_{\Cl D_M}$ is contained in
    $\preccurlyeq_{\Cl C_M}$.
  \item\label{item:admissible-sets-vs-quasi-orders-2} The quasi-orders
    $\preccurlyeq_{\Cl C_M}$ and $\preccurlyeq_{[\Cl C]}$ coincide.
  \item\label{item:admissible-sets-vs-quasi-orders-3} The preimages
    under~$\iota_{\Cl B_{\mathrm{max}}^T}$ of the clopen
    $\preccurlyeq_{\Cl C_M}$-upsets are precisely the members of~$[\Cl
    C]$.
  \item\label{item:admissible-sets-vs-quasi-orders-4} The inclusion
    $\Cl C\subseteq[\Cl D]$ holds if and only if $\preccurlyeq_{\Cl
      D_M}$ is contained in $\preccurlyeq_{\Cl C_M}$.
  \end{enumerate}
\end{Prop}

\begin{proof}
  (\ref{item:admissible-sets-vs-quasi-orders-1}) Suppose that $\Cl
  C\subseteq\Cl D$. Let $u,v\in(\Cl B_{\mathrm{max}}^T)^\star$ be such
  that $u\preccurlyeq_{\Cl D_M}v$ and let $L\in\Cl C_M\cap u$. Since
  $\Cl C_M\subseteq\Cl D_M$, it follows that $L\in v$. This shows that
  $u\preccurlyeq_{\Cl C_M}v$.
  
  (\ref{item:admissible-sets-vs-quasi-orders-2}) Since $[\Cl C]$ is
  $M$-closed, by~(\ref{item:admissible-sets-vs-quasi-orders-1}) we
  conclude that $\preccurlyeq_{[\Cl C]}$ is contained
  in~$\preccurlyeq_{\Cl C_M}$. For the reverse inclusion suppose that
  $u,v\in(\Cl B_{\mathrm{max}}^T)^\star$ are such that
  $u\preccurlyeq_{\Cl C_M}v$ and let $L\in[\Cl C]\cap u$. Then $L$ is
  a finite union of finite intersections of members of~$\Cl C_M$.
  Since $u$ is an ultrafilter, one such intersection
  $L_1\cap\cdots\cap L_n$ (with the $L_i\in\Cl C_M$) belongs to $u$,
  so that each $L_i$ belongs to~$u$. As $u\preccurlyeq_{\Cl C_M}v$, it
  follows that each $L_i$ also belongs to $v$ and, therefore so do
  $L_1\cap\cdots\cap L_n$ and the larger $L$. This shows that
  $u\preccurlyeq_{[\Cl C]}v$.

  (\ref{item:admissible-sets-vs-quasi-orders-3}) Let $L$ be an
  arbitrary admissible subset of the topological algebra~$T$, so that
  $L=(\iota_{\Cl B_{\mathrm{max}}^T})^{-1}(\Cl U_L)$. Suppose first
  that $L\in[\Cl C]$. If $u\preccurlyeq_{\Cl C_M}v$ and $u\in\Cl U_L$,
  then $L\in u$ so that,
  by~(\ref{item:admissible-sets-vs-quasi-orders-2}), $L\in v$ and
  $v\in\Cl U_L$. This shows that $\Cl U_L$ is an $\preccurlyeq_{\Cl
    C_M}$-upset. Conversely, suppose that $\Cl U_L$ is an
  $\preccurlyeq_{\Cl C_M}$-upset, that is, $u\preccurlyeq_{\Cl C_M}v$
  and $L\in u$ implies $L\in v$. As we wish to show that $L\in[\Cl
  C]$, we may as well assume that both $\Cl U_L$ and its complement
  are nonempty. Given $v\in(\Cl U_L)^\complement$, each $u\in\Cl U_L$
  is such that $u\not\preccurlyeq_{\Cl C_M}v$ and so there exists
  $K_{u,v}\in\Cl C_M$ such that $K_{u,v}\in u\setminus v$. As $\Cl
  U_L$ is compact and the open sets $\Cl U_{K_{u,v}}$ with $u\in\Cl
  U_L$ cover $\Cl U_L$, there exist $u_1,\ldots,u_m\in\Cl U_L$ such
  that $\Cl U_L\subseteq\bigcup_{i=1}^m\Cl U_{K_{u_i,v}}=\Cl U_{K_v}$,
  where $K_v=\bigcup_{i=1}^mK_{u_i,v}$ and $v\in(\Cl
  U_{K_v})^\complement$. Since $(\Cl U_L)^\complement$ is also
  compact, there exist $v_1,\ldots,v_n$ such that
  \begin{displaymath}
    (\Cl U_L)^\complement
    =\bigcup_{j=1}^n(\Cl U_{K_{v_j}})^\complement
    =\biggl(\bigcap_{j=1}^n\Cl U_{K_{v_j}}\biggr)^\complement
    =\left(\Cl U_{\bigcap_{j=1}^nK_{v_j}}\right)^\complement.
  \end{displaymath}
  Hence
  $L=\bigcap_{j=1}^nK_{v_j}=\bigcap_{j=1}^n\bigcup_{i=1}^mK_{u_i,v_j}$
  belongs to~$[\Cl C]$.

  (\ref{item:admissible-sets-vs-quasi-orders-4}) The direct
  implication follows from
  (\ref{item:admissible-sets-vs-quasi-orders-1}) and
  (\ref{item:admissible-sets-vs-quasi-orders-2}). For the converse,
  assume that $\preccurlyeq_{\Cl D_M}$ is contained in
  $\preccurlyeq_{\Cl C_M}$ and let $L\in\Cl C_M$. As $\Cl U_L$ is an
  $\preccurlyeq_{\Cl C_M}$-upset, it is also an $\preccurlyeq_{\Cl
    D_M}$-upset. Hence $L\in[\Cl D]$
  by~(\ref{item:admissible-sets-vs-quasi-orders-3}).
\end{proof}

\begin{Cor}
  \label{c:duality}
  Let $\Cl C$ and $\Cl D$ be subsets of~$\Cl B_{\mathrm{max}}^T$. Then
  there is a (unique, onto) continuous order-preserving homomorphism
  $\varphi:\Stsynt(\Cl C)\to\Stsynt(\Cl D)$ such that
  $\varphi\circ\eta_{\Cl C}=\eta_{\Cl D}$ if and only if $\Cl
  D\subseteq[\Cl C]$.
\end{Cor}

\begin{proof} 
  In view of Corollary~\ref{c:admissible-to-po}, the existence of such
  a mapping $\varphi$ is equivalent to $\preccurlyeq_{\Cl C_M}$ being
  contained in~$\preccurlyeq_{\Cl D_M}$, which in turn is equivalent to
  $\Cl D\subseteq[\Cl C]$ by
  Proposition~\ref{p:admissible-sets-vs-quasi-orders}.
\end{proof}

\begin{Prop}
  \label{p:M-closed-0,1-sublattice-vs-preimage}
  Let $\varphi:T\to S$ be a Stone completion and let $\Cl L$ be a
  subset of~$\Cl P_{co}(S)$. Then the following equivalences hold:
  \begin{enumerate}
  \item\label{item:M-closed-0,1-sublattice-vs-preimage-1} $\Cl L$ is a
    $(0,1)$-sublattice of $\Cl P_{co}(S)$ if and only if
    $\varphi^{-1}(\Cl L)$ is a $(0,1)$-sublattice of~$\Cl P_{co}(T)$;
  \item\label{item:M-closed-0,1-sublattice-vs-preimage-2}
    $\Cl L$ is an M-closed subset of $\Cl P_{co}(S)$ if and only if
    $\varphi^{-1}(\Cl L)$ is an M-closed subset of~$\Cl P_{co}(T)$.
  \end{enumerate}
\end{Prop}

\begin{proof}
  As the set mapping $\varphi^{-1}$ restricts to an injective Boolean
  algebra homomorphism $\Cl P_{co}(S)\to\Cl P_{co}(T)$, we see that
  $\Cl L$ is a $(0,1)$-sublattice of $\Cl P_{co}(S)$ if and only if
  $\varphi^{-1}(\Cl L)$ is a $(0,1)$-sublattice of~$\Cl P_{co}(T)$.
  This proves (\ref{item:M-closed-0,1-sublattice-vs-preimage-1}) while
  (\ref{item:M-closed-0,1-sublattice-vs-preimage-2})~is given by
  Proposition~\ref{p:M-closed-in-Stone-completion}.
\end{proof}

The following result gives the structure of all Priestley stamps.

\begin{Prop}
  \label{p:stamps} 
  Let $\varphi:T\to S$ be a stamp and consider the subset $\Cl
  C=\varphi^{-1}\bigl(\Cl P_{uco}(S)\bigr)$ of\/~$\Cl
  B_{\mathrm{max}}^T$. Then $\Cl C$ is an M-closed sublattice of~$\Cl
  B_{\mathrm{max}}^T$ and there is a (unique) homomorphism
  $\delta:(S,{\le})\to(\Stsynt(\Cl C),{\le_{\Cl C}})$ of ordered
  topological algebras such that $\delta\circ\varphi=\eta_{\Cl C}$.
  Moreover, the following hold:
  \begin{enumerate}
  \item\label{item:stamps-1} the lattice $\Cl P_{uco}(S)$ generates
    the Boolean algebra $\Cl P_{co}(S)$ if and only if $\delta$ is an
    isomorphism of topological algebras;
  \item\label{item:stamps-2} the pair $(S,{\le})$ is a Priestley
    algebra if and only if $\delta$ is an isomorphism of ordered
    topological algebras.
  \end{enumerate}
\end{Prop}

\begin{proof}
  Since $\le$ is a stable relation on~$S$, if $s_1\le s_2$ in~$S$ and
  $g\in M(S)$ then $g(s_1)\le g(s_2)$. It follows that $\Cl
  P_{uco}(S)$ is $M$-closed, while this set is obviously a
  $(0,1)$-sublattice of~$\Cl P_{co}(S)$. By
  Proposition~\ref{p:M-closed-0,1-sublattice-vs-preimage}, as $\Cl C$
  is contained in the $\Omega$-finite Boolean algebra
  $\varphi^{-1}(\Cl P_{co}(S))$ (cf.\
  \cite[Corollary~6.10]{Almeida&Klima:2024a}), it follows that $\Cl C$
  is an M-closed $(0,1)$-sublattice of~$\Cl B_{\mathrm{max}}^T$,
  whence $\Cl C=[\Cl C]$. Thanks to
  \cite[Corollary~7.4]{Almeida&Klima:2024a}, there is a continuous
  homomorphism $\delta:S\to \langle\Cl C\rangle^\star=\Stsynt(\Cl C)$
  of Stone topological algebras such that
  $\delta\circ\varphi=\eta_{\Cl C}$. As the set function
  $\varphi^{-1}$ restricts to a Boolean algebra isomorphism $\Cl
  P_{co}(S)\to\langle\Cl C\rangle$, by the same
  \cite[Corollary~7.4]{Almeida&Klima:2024a}, we obtain the
  equivalence~(\ref{item:stamps-1}). It remains to establish the
  equivalence~(\ref{item:stamps-2}) where the reverse implication is
  given by Corollary~\ref{c:admissible-to-po}. By
  Proposition~\ref{p:admissible-sets-vs-quasi-orders}(\ref{item:admissible-sets-vs-quasi-orders-3}),
  we know that $\eta_{\Cl C}^{-1}\bigl(\Cl P_{uco}(\Stsynt(\Cl
  C))\bigr)=[\Cl C]=\Cl C$. Since, in a Priestley algebra, the clopen
  upsets completely determine the order, the direct implication
  in~(\ref{item:stamps-2}) also holds.
\end{proof}

In the special case where $\Cl C$ consists of a single admissible
subset of~$T$, we denote $(\Stsynt(\Cl C),{\le_{\Cl C}})$ by
$(\Stsynt(L),{\le_L})$, or $(\Stsynt_T(L),{\le_L^T})$ if it is
convenient to refer to the topological algebra $T$, and call it the
\emph{Priestley syntactical algebra of~$L$}. We may further add
reference to the signature, $(\Stsynt^\Omega(L),{\le_L^\Omega})$, if
more than one signature may be involved (as in
Section~\ref{sec:min-automata}). We also denote $\eta_{\Cl C}$ by
$\eta_L$ (or $\eta_L^\Omega$) and call it the \emph{Priestley
  completion of $T$ determined by~$L$}.

Recall that, for a subset $L$ of~$T$, the \emph{algebraic syntactic
  quasi-order} on $T$ is defined by
\begin{equation}
  \label{eq:algebraic-syntactic-qo}
  s\curlyeqprec_L t
  \quad\text{if }
  \forall f\in M(T)\ \bigl( s\in f^{-1}(L)\implies t\in f^{-1}(L)\bigr).
\end{equation}
As above, we may also denote this relation by $\curlyeqprec_L^\Omega$
when the signature considered is~$\Omega$. The relation
$\curlyeqprec_L$ is the largest stable quasi-order on $T$ for which
$L$ is an \emph{upset}. The associated equivalence relation
${\curlyeqprec_L}\cap{\curlyeqsucc_L}$ on $T$ is the \emph{syntactic
  equivalence} $\sigma_L$. The quotient algebra ordered algebra
$(T/\sigma_L,{\le})$ is called the \emph{ordered syntactic algebra}
of~$L$ and it is denoted $\synt(L)$ or $\synt_T(L)$, as may be more
convenient.

\begin{Prop}
  \label{p:syntactic-in-dual}
  Let $L$ be an admissible subset of\/~$T$. Then, for $s,t\in T$, the
  inequality $s\curlyeqprec_Lt$ holds if and only if so does
  $\eta_L(s)\le_L\eta_L(t)$. Hence, the congruence $\ker\eta_L$ is the
  syntactic congruence $\sigma_L$ of $L$ in $T$ and so the image
  $\eta_L(T)$ is isomorphic with the syntactic algebra $\synt(L)$. In
  particular, the syntactic algebra $\synt(L)$ has a natural structure
  of topological algebra, where a basis of the topology is given by
  all sets of the form $K/\sigma_L$ with $K\in\langle L\rangle$.
\end{Prop}

\begin{proof}
  By definition of $\eta_L$ and
  Corollary~\ref{c:B(L)}, the following equivalences hold for all
  $s,t\in T$:
  \begin{align*}
    \eta_L(s)\le_L\eta_L(t)
    &\iff
      \forall K\in\langle L\rangle\ \bigl(s\in K \implies t\in K\bigr) \\
    &\iff
      \forall f\in M(T)\ \bigl(s\in f^{-1}(L) \implies t\in f^{-1}(L)\bigr) \\
    &\iff
      s \curlyeqprec_L t,
  \end{align*}
  which entails the equality $\ker\eta_L=\sigma_L$. All that remains
  to observe is that the induced topology on $\eta_L(T)$ admits the
  basis consisting of the sets of the form $\eta_L(K)=\Cl
  U_K\cap\eta_L(T)$ with $K\in\langle L\rangle$.
\end{proof}

The following result  gives yet another connection of $\Stsynt(L)$ with
syntactic algebras.

\begin{Prop}
  \label{p:Stsynt-as-styntactic-algebra}
  Let $T$ be a topological algebra, let $\Cl B=\Cl
  B_{\mathrm{max}}^T$, and let $L$ be an admissible subset of~$T$. Let
  $\delta:\Cl B^\star\to\Stsynt(L)$ be the continuous homomorphism
  such that $\delta\circ\iota_{\Cl B}=\eta_L$. Then the quasi-order
  $(\delta\times\delta)^{-1}({\le_L})$ coincides with the algebraic
  syntactic quasi-order of~$K=\Cl U^{\Cl B}_L$, so that the congruence
  $\ker\delta$ is the syntactic congruence of~$K$.
\end{Prop}

\begin{proof}
  By \cite[Corollary~7.4]{Almeida&Klima:2024a}, there is indeed a
  continuous homomorphism $\delta$ such that the following diagram
  commutes, namely $\delta(u)=u\cap\langle L\rangle$:
  \begin{displaymath}
    \xymatrix{
      T
      \ar[r]^{\iota_{\Cl B}}
      \ar[d]_{\iota_{\langle L\rangle}}
      &
      \Cl B^\star
      \ar[ld]^\delta
      \\
      \Stsynt(L)
      &
    }
  \end{displaymath}
  Given $u,v\in\Cl B^\star$, the following equivalences hold:
  \begin{align*}
    u\curlyeqprec_K v
    &\iff
      \forall h\in M(\Cl B^\star)\ \left(u\in h^{-1}(K) \implies v\in
      h^{-1}(K)\right)\\
    &\iff
      \forall g\in M_{\iota_{\Cl B}}(\Cl B^\star)\ \left(u\in g^{-1}(K) \implies v\in
      g^{-1}(K)\right)\\
    &\iff
      \forall f\in M(T)\ \left(f^{-1}(L)\in u \implies f^{-1}(L)\in v\right)\\
    &\iff
      u\preccurlyeq_{\{L\}_M}v\\
    &\iff
      \delta(u)\le_L\delta(v),
  \end{align*}
  where the steps are justified respectively by the definition of
  $\curlyeqprec_K$,
  Lemma~\ref{l:reduction-to-induced-linear-transformations},
  Lemma~\ref{l:induced-linear-in-Stsynt}, and the definitions of
  $\preccurlyeq_{\{L\}_M}$ 
  and $\le_L$. The final conclusion follows from the
  fact that $\sigma_K$ is the intersection ${\curlyeqprec_K}\cap{\curlyeqsucc_K}$,
  which coincides with the kernel of~$\delta$ because of the above and
  the fact that $\le_L$ is a partial order.
\end{proof}

The next result exhibits a minimality property of the Priestley
completion

\begin{Prop}
  \label{p:ordered-recognizers}
  Let $\varphi:T\to S$ be a stamp where $\le$ denotes the order of~$S$
  and let $K\in\Cl P_{uco}(S)$ and $L=\varphi^{-1}(K)$. Then $\langle
  L\rangle_T$ is contained in $\varphi^{-1}\bigl(\Cl P_{co}(S)\bigr)$
  and the unique continuous homomorphism $\varepsilon:S\to\Stsynt(L)$
  such that $\varepsilon\circ\varphi=\eta_L$ preserves order.
\end{Prop}

\begin{proof}
  By assumption, the set $L$ belongs to the $\Omega$-finite Boolean
  subalgebra~$\varphi^{-1}\bigl(\Cl P_{co}(S)\bigr)$ of $\Cl
  P_{co}(T)$, from which we obtain the inclusion $\langle
  L\rangle_T\subseteq\varphi^{-1}\bigl(\Cl P_{co}(S)\bigr)$.
  
  Let $\Cl C=\varphi^{-1}\bigl(\Cl P_{uco}(S)\bigr)$ and let $\Cl
  D=\{L\}$. By Corollary~\ref{c:duality}, there is a continuous
  order-preserving homomorphism $\psi:\Stsynt(\Cl C)\to\Stsynt(\Cl
  D)=\Stsynt(L)$
  such that $\psi\circ\eta_{\Cl C}=\eta_L$. By
  Proposition~\ref{p:stamps}, there is a continuous order-preserving
  homomorphism $\delta:(S,{\le})\to(\Stsynt(\Cl C),{\le_{\Cl
      C}})$ such that $\delta\circ\varphi=\eta_{\Cl C}$, as in the
  following commutative diagram:
  \begin{displaymath}
    \xymatrix{
      T
      \ar[r]^\varphi
      \ar[rd]_{\eta_{\Cl C}}
      \ar[d]_{\eta_L}
      &
      S
      \ar[d]^\delta
      \\
      \Stsynt(L)
      &
      \Stsynt(\Cl C)
      \ar[l]^\psi
    }
  \end{displaymath}
  To complete
  the proof, it suffices to take $\varepsilon=\psi\circ\delta$.
\end{proof}

Let $\mathbf{OSC}_\Omega$ be the category of ordered Stone completions
of topological $\Omega$-algebras:
\begin{itemize}
\item the objects are the continuous homomorphisms $\varphi:U\to S$,
  where $U$ is a topological $\Omega$-algebra and $S$ is an ordered
  Stone topological $\Omega$-algebra  such that $\varphi(U)$ is dense
  in~$S$;
\item the morphisms $(\varphi:U\to S)\to(\varphi':U'\to S')$ are pairs
  $(\psi,\delta)$ such that $\varphi'\circ\psi=\delta\circ\varphi$,
  where $\psi:U\to U'$ and $\delta:S\to S'$ are continuous
  homomorphisms and $\delta$ preserves the order.
\end{itemize}
The subcategory $\mathbf{OSC}_\Omega(T)$ consists of the objects of
$\mathbf{OSC}_\Omega$ of the form $T\to S$ and whose morphisms $(T\to
S)\to(T\to S')$ have $\mathrm{id}_T$ as first component.

\begin{Cor}
  \label{c:St-C-vs-St-L}
  Let $\Cl C$ be a subset of~$\Cl B_{\mathrm{max}}^T$. Then the
  ordered Stone completion $\eta_{\Cl C}:T\to\Stsynt(\Cl C)$ is the
  product in the category $\mathbf{OSC}_\Omega(T)$ of the ordered
  syntactical Stone completions $\eta_L:T\to\Stsynt(L)$ with $L\in\Cl
  C$.
\end{Cor}

\begin{proof}
  By Proposition~\ref{p:ordered-recognizers}, for each $L\in\Cl C$ we
  have in $\mathbf{OSC}_\Omega(T)$ a morphism
  $(\mathrm{id}_T,\pi_L):\eta_{\Cl C}\to\eta_L$, which means that
  $\pi_L$ is an order preserving continuous homomorphism such that the
  following diagram commutes:
  \begin{equation}
    \label{eq:St-C-vs-St-L-1}
    \begin{split}
    \xymatrix{
     T
     \ar[d]^{\eta_L}
     &
     T
     \ar[l]_{\mathrm{id}_T}
     \ar[d]^{\eta_{\Cl C}}
     \\
     \Stsynt(L)
     &
     \Stsynt(\Cl C)
     \ar[l]_{\pi_L}
    }      
    \end{split}
  \end{equation}
  Note that, for $u,v\in\Stsynt(\Cl C)$, the following equivalence
  holds:
  \begin{displaymath}
    u\le_{\Cl C}v
    \iff
    \forall L\in\Cl C,\ \pi_L(u)\le_L\pi_L(v). 
  \end{displaymath}
  Hence, the mapping
  \begin{align*}
    \pi:\Stsynt(\Cl C)&\to\prod_{L\in\Cl C}\Stsynt(L)\\
    u&\mapsto \bigl(\pi_L(u)\bigr)_{L\in\Cl C}
  \end{align*}
  is an embedding of ordered Stone topological algebras. Since
  Diagram~(\ref{eq:St-C-vs-St-L-1}) commutes for every $L\in\Cl C$,
  the image of~$\pi$ is the closure of the set
  \begin{displaymath}
    X
    =\pi\bigl(\eta_{\Cl C}(T)\bigr)
    =\left\{\bigl(\eta_L(t)\bigr)_{L\in\Cl C}: t\in T\right\}.
  \end{displaymath}
  
  Suppose that $\varphi:T\to S$ is an ordered Stone completion and,
  for each $L\in\Cl C$, $(\mathrm{id}_T,\delta_L)$ is a morphism
  $\varphi\to\eta_L$. Then $\delta(s)=\bigl(\delta_L(s)\bigr)_{L\in\Cl
    C}$ defines an order preserving continuous homomorphism
  $\delta:S\to\prod_{L\in\Cl C}\Stsynt(L)$. Since
  $\delta_L\circ\varphi=\eta_L$ for each $L\in\Cl C$ and the image
  of~$\varphi$ is dense in~$S$, it follows that the image of $\delta$
  coincides with the image of~$\pi$. This leads to the following
  commutative diagram:
  \begin{displaymath}
    \xymatrix{
      T
      \ar[dd]_\varphi
      \ar[r]^{\mathrm{id}_T}
      &
      T
      \ar[d]_{\eta_L}
      &
      T
      \ar[l]_{\mathrm{id}_T}
      \ar[r]^{\mathrm{id}_T}
      \ar[d]_{\eta_{\Cl C}}
      &
      T
      \ar[d]_{\pi\circ\eta_{\Cl C}}
      \\
      &
      \Stsynt(L)
      &
      \Stsynt(\Cl C)
      \ar[l]_{\pi_L}
      \ar[r]^(.55)\pi_(.55){\simeq}
      &
      \overline{X}
      \ar@{^{(((}->}[d]
      \\
      S
      \ar[ru]^(.6){\delta_L}
      \ar@/_2.5mm/[rru]^(.6){\pi^{-1}\circ\delta}
      \ar[rrr]^\delta
      &&&
      \prod_{L\in\Cl C}\Stsynt(L)
    }      
  \end{displaymath}
  Thus the morphisms $(\mathrm{id}_T,\pi_L)$ ($L\in\Cl C$) have the
  required universal property to define a product in the category
  $\mathbf{OSC}_\Omega(T)$.
\end{proof}

\section{An extension of Eilenberg's correspondence}
\label{sec:Eilenberg}

Let $\mathfrak{C}$ be a subcategory of the category of all topological
$\Omega$-algebras where morphisms are continuous homomorphisms between
them. Whenever we have some morphisms in the category $\mathfrak{C}$,
this leads to extra conditions on certain classes we consider in this
section and which are parameterized by the category $\mathfrak{C}$.
For that reason, we emphasize the case when there are no other
morphisms than those necessarily needed, thus we denote
$\mathfrak{C}_0$ the subcategory of $\mathfrak{C}$ with the same
objects and identities on the objects as unique morphisms.

A \emph{$\mathfrak{C}$-prevariety of languages} is a contravariant
functor \Cl V from $\mathfrak{C}$ to the category of $(0,1)$-lattices
such that:
\begin{enumerate}
\item\label{item:def-P-variety-languages-1} for each object $T$
  in~$\mathfrak{C}$, $\Cl V(T)$ is a $(0,1)$-sublattice 
  of~$\Cl B_{\mathrm{max}}^T$;
\item\label{item:def-P-variety-languages-2} for each morphism
  $\psi\in\mathfrak{C}(T,T')$, the corresponding $(0,1)$-lattice 
  homomorphism $\Cl V(\psi):\Cl V(T')\to\Cl V(T)$ is the mapping
  $\Cl V(\psi)=\psi^{-1}|_{\Cl V(T')}$.
\end{enumerate}
The second condition means that the choices of $\Cl V(T')$ and $\Cl
V(T)$ have to be done in such a way that for every
$\psi\in\mathfrak{C}(T,T')$ and $L\in \Cl V(T')$ we have
$\psi^{-1}(L)\in \Cl V(T)$. Thus, we say informally that $\Cl V$ is
closed under morphism preimages. Moreover, the mapping $\Cl V(\psi)$
is always a $(0,1)$-lattice homomorphism if it is defined, and the
condition~(\ref{item:def-P-variety-languages-2}) immediately implies
that the considered \Cl V is indeed a functor. Altogether, to see that
a given $\Cl V$ is a $\mathfrak{C}$-variety of languages, one needs to
check that every $\Cl V(T)$ is a $(0,1)$-sublattice of~$\Cl
B_{\mathrm{max}}^T$ and that the system $(\Cl V(T))_{T\in
  \mathfrak{C}}$ is closed under preimages for morphisms
in~$\mathfrak{C}$. Finally, the \emph{$\mathfrak{C}_0$-prevariety of
  languages} is just a collection $(\Cl V(T))_{T\in\mathfrak{C}_0}$ of
$(0,1)$-sublattices of the corresponding Boolean algebras~$\Cl
B_{\mathrm{max}}^T$. By the definition, the restriction
to~$\mathfrak{C}_0$ of every $\mathfrak{C}$-prevariety of languages is
also a $\mathfrak{C}_0$-prevariety of languages.

We denote by $\mathbb{L}_0^{\mathfrak{C}}$ the class of all
$\mathfrak{C}_0$-prevarieties of languages, and by
$\mathbb{L}^{\mathfrak{C}}$ its subclass of all
$\mathfrak{C}$-prevarieties of languages.

Paraphrasing \cite{Pin&Straubing:2005}, we call an object $\varphi$ of
$\osc$ a \emph{stamp}, and if $\varphi : T \to S$ for
$T\in\mathfrak{C}$ then we call $\varphi$ a
\emph{$\mathfrak{C}$-stamp}. Recall that, in this situation, $T$ is a
topological algebra, $(S,\le)$ is an ordered Stone topological
algebra, and $\varphi$ is a continuous homomorphism such that
$\varphi(T)$ is dense subset of $S$. Moreover, we call a stamp
$\varphi : T \to S$ a \emph{Priestley stamp} if the ordered Stone
topological algebra $(S,\le)$ satisfies the Priestley condition. For
$\Cl C\subseteq \Cl B_{\mathrm{max}}^{T}$, that is a set of admissible
languages over $T\in \mathfrak{C}$, we call the continuous
homomorphism $\eta_{\Cl C}: T \to \Stsynt(\Cl C)$ into the ordered
Stone syntactical algebra $(\Stsynt(\Cl C),{\le_{\Cl C}})$ of $\Cl C$
the \emph{syntactical stamp} of~$\Cl C$. Every syntactical stamp is a
Priestley stamp by Corollary~\ref{c:admissible-to-po}.

Given a class \Cl S of $\mathfrak{C}$-stamps, consider the following
associated classes:
\begin{itemize}
\item the class $H\Cl S$ consists of all $\mathfrak{C}$-stamps of the
  form $\alpha\circ\varphi: T \to S$ where $\varphi: T \to S'$ belongs
  to~\Cl S and $\alpha:(S',\le)\to (S,\le)$ is an onto morphism
  between ordered Stone topological algebras -- in this case we call
  the stamp $\alpha\circ\varphi$ a \emph{homomorphic image of the
    stamp $\varphi$};
\item the class $P\Cl S$ consists of all $\mathfrak{C}$-stamps of the
  form $\prod_{i\in I}\varphi_i: T \to
  \overline{\mathrm{Im}(\prod_{i\in I}\varphi_i)}$, where 
  $\varphi_i: T\to S_i$ belongs to~\Cl S with $(S_i,\le)$ 
  being an ordered Stone topological algebra (for every $i\in I$),
  and $(\prod_{i\in I}\varphi_i)(w)=(\varphi_i(w))_{i\in I}\in 
  \prod_{i\in I}S_i$ for $w\in T$ -- we call the stamp 
  $\prod_{i\in I}\varphi_i$ the \emph{product} of the stamps $\varphi_i$;
\item the class $S_{\mathfrak{C}}\Cl S$ consists of all
  $\mathfrak{C}$-stamps of the form $\varphi\circ\psi: T \to
  \overline{\mathrm{Im} (\varphi\circ\psi)}$ where $\varphi: T'\to S$
  is a member of~\Cl S, $(S,\le)$ is an ordered Stone topological
  algebra, $\psi\in\mathfrak{C}(T,T')$, and
  $(\overline{\mathrm{Im}(\varphi\circ\psi)},\le)$ is viewed as a
  subalgebra of $S$ with the induced order.
\end{itemize}
A \emph{$\mathfrak{C}_0$-prevariety of stamps} is a class $\vst{V}$ of
$\mathfrak{C}$-stamps that is closed under the operators $H$ and $P$. 
A \emph{$\mathfrak{C}$-prevariety of stamps} is a $\mathfrak{C}_0$-prevariety 
of stamps which is also closed under the operator $S_{\mathfrak{C}}$.

We denote by $\mathbb S_0^{\mathfrak{C}}$ the class of all 
$\mathfrak{C}_0$-prevarieties of stamps and by $\mathbb S^{\mathfrak{C}}$ 
its subclass formed by all $\mathfrak{C}$-prevarieties of stamps. 

Let $\vst{P}_{\mathfrak{C}}$ denote the class of all Priestley
$\mathfrak{C}$-stamps. Notice that $\vst{P}_{\mathfrak{C}}$ is not a
$\mathfrak{C}_0$-prevariety because it need not be closed under
homomorphic images. Nevertheless, it is closed under the other two 
fundamental operators as the next lemma states.

\begin{Lemma} 
  \label{l:closure-properties-of-priestley-stamps}
  The class $\vst P_{\mathfrak{C}}$ 
  is closed under the operators $P$ and $S_{\mathfrak{C}}$. 
\end{Lemma} 
\begin{proof}
  For a stamp $\varphi : T \to S$, the Priestley condition 
  is solely a property of the ordered Stone topological algebra $(S,\le)$. 
  Thus the class $\vst P_{\mathfrak{C}}$ is closed under the operator 
  $S_\mathfrak{C}$, because a closed subspace of a Priestley space 
  (with induced order) is a Priestley space.  

  Similarly, the product of Priestley spaces is a Priestley space. 
  Thus the class  $\vst P_{\mathfrak{C}}$ is closed under the operator 
  $P$ which is defined by product and closed subalgebra constructions.
\end{proof}

We would like to obtain an appropriate formula for the smallest class
closed under the above operators which contains a given class. This is
done using the usual pointwise order between operators: given two
operators $O$ and $R$ on classes of $\mathfrak{C}$-stamps, we write
$O\le R$ if $O\Cl S\subseteq R\Cl S$ for every class \Cl S of
$\mathfrak{C}$-stamps. The operators can be composed in the natural
way, that is $RO\Cl S=R(O\Cl S)$. It is then routine to check the
following inequalities:
\begin{align*}
  S_{\mathfrak{C}}S_{\mathfrak{C}} &\le S_{\mathfrak{C}}
  &                                     
  S_{\mathfrak{C}}H&\le HS_{\mathfrak{C}}
  &
  S_{\mathfrak{C}}P&\le PS_{\mathfrak{C}}\\
  HH&\le H
  &
  PP&\le HP
  &
  PH &\le HP.
\end{align*}
These observations yield the following result.

\begin{Lemma}
  \label{l:C-variety}
  The smallest $\mathfrak{C}_0$-prevariety of stamps containing a given 
  class \Cl S of $\mathfrak{C}$-stamps is the class $HP\Cl S$.
  The smallest $\mathfrak{C}$-prevariety of stamps containing a given 
  class \Cl S of $\mathfrak{C}$-stamps is the class $HPS_{\mathfrak{C}}\Cl S$.
\end{Lemma}

Notice, that for each $T\in \mathfrak{C}$, an arbitrary 
$\mathfrak{C}_0$-prevariety of stamps contains the \emph{trivial stamp} 
$\varphi : T \to \mathbf{1}$, where $(\mathbf{1},=)$ is the one-element 
(ordered) Stone topological algebra, since that may be obtained by the 
operator $P$, when the considered index set $I$ is empty. 

We are interested in the class of languages over $T\in\mathfrak{C}$ which 
may be recognized by stamps from a given $\mathfrak{C}_0$-prevariety $\vst V$ 
of stamps: 
\begin{displaymath}
  \Rec_{\vst V}(T)=\{ L \in \Cl B_{\mathrm{max}}^T ;
  \eta_L \in \vst V\}.
\end{displaymath}
From Proposition~\ref{p:ordered-recognizers} we get the following 
alternative description of~$\Rec_{\vst V}(T)$.

\begin{Lemma}
  \label{l:variety-of-stamps-recognition}
  Let $\vst V$ be a $\mathfrak{C}_0$-prevariety of stamps and $T\in\mathfrak{C}$. 
  An admissible language $L\in \Cl B_{max}^T$ over $T$ belongs to 
  $\Rec_{\vst V}(T)$ if and only if there exist $(\varphi : T \to S) \in \vst V$, 
  where $(S,\le)$ is an ordered Stone topological algebra, and a clopen upset 
  $K$ in $S$ such that $L=\varphi^{-1}(K)$.
\end{Lemma}

\begin{Remark}
  \label{r:Rec-is-a-sublattice}
  We claim that $\Rec_{\vst V}(T)$ is a $(0,1)$-sublattice of $\Cl
  B_{\mathrm{max}}^T$. Indeed, as mentioned above, every
  $\mathfrak{C}_0$-prevariety of stamps $\vst V$ contains the trivial
  stamps $T\to\mathbf{1}$ with $T\in\mathfrak{C}$, which are the
  syntactical stamps of the languages $\emptyset,T\in\Cl
  B_{\mathrm{max}}^T$. For $L,K\in\Rec_{\vst V}(T)$, the languages
  $L\cap K$ and $L\cup K$ are recognized by the stamp $\eta_L \times
  \eta_K$, which belongs to~$\vst V$. Finally, we get that $L\cap K$
  and $L\cup K$ belong to $\Rec_{\vst V}(T)$ by
  Lemma~\ref{l:variety-of-stamps-recognition}.
\end{Remark}

We define mappings $\Phi:\mathbb{L}^{\mathfrak{C}}_0\to\mathbb{S}^{\mathfrak{C}}_0$ 
and $\Psi:\mathbb{S}^{\mathfrak{C}}_0\to\mathbb{L}^{\mathfrak{C}}_0$ as follows:
\begin{itemize}
\item for a $\mathfrak{C}_0$-prevariety of languages \Cl V
  in~$\mathbb{L}^{\mathfrak{C}}_0$, we let $\Phi(\Cl V)$ be the
  $\mathfrak{C}_0$-prevariety of stamps generated by the syntactic
  stamps of the sets $\Cl V(T)$, that is,
  \begin{displaymath}
    \Phi(\Cl  V)=H \bigl( \{ \eta_{\Cl V(T)} : T \to \Stsynt(\Cl V(T)) ; T\in
    \mathfrak{C}\} \bigr);
  \end{displaymath}
  \item for a $\mathfrak{C}_0$-prevariety of stamps $\vst{V}$, let $\Psi(\vst{V})$ 
    be the $\mathfrak{C}_0$-prevariety of all languages recognized by stamps 
    from $\vst{V}$, that is $ \Psi(\vst{V}) (T)=\Rec_{\vst V} (T)$.
\end{itemize}
By Remark~\ref{r:Rec-is-a-sublattice}, the mapping $\Psi$ is correctly defined,
however the definition of $\Phi$ needs a short comment. The class 
$\Cl S=\{\eta_{\Cl V(T)}; T\in \mathfrak{C}\}$ contains a single stamp of the 
form $\varphi : T \to S$ for every $T\in \mathfrak{C}$. Thus the class $P\Cl S$ 
consists of isomorphic copies of the stamps in $\Cl S$ and the trivial stamps. 
We deduce that $P\Cl S \subseteq H\Cl S$ and get 
$H\Cl S \subseteq HP\Cl S \subseteq HH\Cl S \subseteq  H\Cl S$. Therefore 
$\Phi(\Cl  V)=H\Cl S=HP\Cl S$ is the $\mathfrak{C}_0$-prevariety of stamps 
generated by the syntactic stamps of $\Cl V(T)$ for $T\in\mathfrak{C}$ as 
stated in the definition.

Our aim is to show that the pair $(\Phi,\Psi)$ of mappings establishes
a monotone Galois connection between $\mathbb{L}^{\mathfrak{C}}_0$ and
$\mathbb{S}^{\mathfrak{C}}_0$. Here the classes
$\mathbb{L}^{\mathfrak{C}}_0$ and $\mathbb{S}^{\mathfrak{C}}_0$ are
naturally ordered by the (pointwise) inclusion order. In particular,
for a pair of $\mathfrak{C}$-prevarieties of languages $\Cl V$ and
$\Cl W$, we write $\Cl V \subseteq \Cl W$ if $\Cl V(T)\subseteq \Cl
W(T)$ holds for every $T\in\mathfrak{C}$.

\begin{Prop}
  \label{p:Galois-connection} 
  The pair of mappings $(\Phi,\Psi)$ is a monotone Galois connection 
  between $\mathbb{L}^{\mathfrak{C}}_0$ and $\mathbb{S}^{\mathfrak{C}}_0$ 
  in the sense that both functions are order-preserving and the following
  equivalence holds for every  $\mathfrak{C}_0$-prevariety of languages 
  $\Cl V$ from $\mathbb{L}^{\mathfrak{C}}_0$ and $\mathfrak{C}_0$-prevariety 
  of stamps $\vst V$ from $\mathbb S^{\mathfrak{C}}_0$:
  \begin{equation}
    \label{eq:Galois-connection}
    \Cl V \subseteq  \Psi({\vst V}) \iff \Phi(\Cl V) \subseteq \vst V.
  \end{equation}
\end{Prop}

\begin{proof} 
  At first, we check that $\Phi$ and $\Psi$ are order-preserving
  functions. Let $\vst V \subseteq \vst W$ be a pair of
  $\mathfrak{C}_0$-prevarieties of stamps, and $T$ be in
  $\mathfrak{C}$. By definition, $\vst V \subseteq \vst W$ implies
  $\Rec_{\vst V}(T) \subseteq\Rec_{\vst W}(T)$, which means $\Psi(\vst
  V)\subseteq \Psi(\vst W)$. Now, assume that $\Cl V \subseteq \Cl W$
  be a pair of $\mathfrak{C}_0$-prevarieties of languages. The
  inclusion $\Phi(\Cl V) \subseteq \Phi(\Cl W)$ is a consequence of
  Corollary~\ref{c:admissible-to-po} and
  Proposition~\ref{p:admissible-sets-vs-quasi-orders}, which imply
  that, for every $T\in \mathfrak{C}$, the stamp $\eta_{\Cl V(T)}$ is
  a homomorphic image of $\eta_{\Cl W (T)}$.

  Now we establish the equivalence~(\ref{eq:Galois-connection}). We
  start with the implication from left to right and assume that $\Cl V
  \subseteq \Psi({\vst V})$. Since $\vst V$ is a
  $\mathfrak{C}_0$-prevariety of stamps, it is enough to show that for
  every $T\in\mathfrak{C}$ the syntactical stamp $\eta_{\Cl V(T)} : T
  \to \Stsynt(\Cl V(T))$ belongs to $\vst V$. For every $L\in \Cl
  V(T)$, by the assumption that $\Cl V \subseteq \Psi({\vst V})$, we
  have $\eta_L\in \vst V$. By Corollary~\ref{c:St-C-vs-St-L}, the
  syntactical stamp $\eta_{\Cl V(T)}$ belongs to $\vst V$ as it is a
  product (in the sense of the operator $P$) of syntactical stamps
  $\eta_L\in \vst V$.

  Conversely, assume that the inclusion $\Phi(\Cl V) \subseteq \vst V$
  holds, and take $T\in \mathfrak{C}$ and $L\in \Cl V(T)$. If we apply
  Corollary~\ref{c:admissible-to-po} to $\Cl C=\Cl V(T)$, we obtain
  that $L \in \Cl C \subseteq \eta_{\Cl C}^{-1} (\Cl P_{uco} (\Stsynt
  (\Cl C)))$, Thus, the language $L$ is recognized by the stamp $\eta_{\Cl V(T)} \in \Phi(\Cl V) \subseteq \vst V$. By
  Lemma~\ref{l:variety-of-stamps-recognition} we get $L\in \Rec_{\vst
    V} (T)=\Psi(\vst V)(T)$.
\end{proof}

From the general theory of Galois connections, it follows from
Proposition~\ref{p:Galois-connection} that
$\Phi\circ\Psi\circ\Phi=\Phi$, $\Psi\circ\Phi\circ\Psi=\Psi$, and that
$\Psi\circ\Phi$ is a closure operator, $\Phi\circ\Psi$ is an interior
operator, and the mappings $\Phi$ and $\Psi$ restrict to mutual
inverse bijections between the sets
$\Psi(\Phi(\mathbb{L}^{\mathfrak{C}}_0))$ and
$\Phi(\Psi(\mathbb{S}^{\mathfrak{C}}_0))$. We proceed to identify
these subsets of $\mathbb{L}^{\mathfrak{C}}_0$ and
$\mathbb{S}^{\mathfrak{C}}_0$, respectively.

\begin{Lemma}
  \label{l:characterisation-of-stable-stamps-prevarieties}
  Let $\vst V$ be a $\mathfrak{C}_0$-prevariety of stamps. Then a
  stamp $\varphi : T \to S$ belongs to $\Phi(\Psi(\vst V))$ if and
  only if it is a homomorphic image of a Priestley stamp $\varphi' : T
  \to S'$ from $\vst V$. In other words,
  \begin{displaymath}
    \Phi(\Psi(\vst V))=H(\vst{P}_{\mathfrak{C}} \cap \vst V).
  \end{displaymath}
\end{Lemma}

\begin{proof}
  We let $\Cl V=\Psi(\vst V)$ and take an arbitrary $\varphi \in
  \Phi(\Psi(\vst V))$. By the definition of $\Phi$, the stamp
  $\varphi$ is a homomorphic image of $\eta_{\Cl V(T)}$, where $\Cl
  V(T)=\Psi(\vst V)(T)=\Rec_{\vst V}(T)$. By
  Corollary~\ref{c:admissible-to-po}, the stamp $\eta_{\Cl V(T)}$ is a
  Priestley stamp. By Corollary~\ref{c:St-C-vs-St-L}, $\eta_{\Cl
    V(T)}$ is a product of the Priestley completions $\eta_L$ where
  $L\in \Rec_{\vst V}(T)$. This means that all $\eta_L$ belong to
  $\vst V$ and, therefore, so does the product $\eta_{\Cl V(T)}$.
  Hence, $\varphi$ is a homomorphic image of the Priestley stamp
  $\eta_{\Cl V(T)}$ in $\vst V$.

  Conversely, let $\varphi'=\alpha \circ \varphi$ be a homomorphic
  image of a Priestley stamp $\varphi : T \to S$ from $\vst V$, where
  $\alpha : (S,\le) \to (S',\le)$ is an order preserving continuous
  homomorphism of ordered Stone topological algebras. We consider $\Cl
  D =\varphi^{-1}(\Cl P_{uco}(S))$. Since $\varphi \in \vst V$, we
  have $\Cl D\subseteq \Rec_{\vst V}(T)=\Psi(\vst V)(T)=\Cl V(T)$.
  This means that $\eta_{\Cl D}$ is a homomorphic image of $\eta_{\Cl
    V(T)}$ by Corollary~\ref{c:duality}. However, the stamp $\eta_{\Cl
    D}$ is isomorphic to the Priestley stamp $\varphi$ by
  Proposition~\ref{p:stamps}. Altogether, $\varphi'$ is a
  homomorphic image of $\eta_{\Cl V(T)} \in \Phi(\Cl V)=\Phi(\Psi(\vst
  V))$. Since $\Phi(\Psi(\vst V))$ is a $\mathfrak{C}_0$-prevariety of
  stamps, we get $\varphi'\in \Phi(\Psi(\vst V))$.
\end{proof}

A $\mathfrak{C}_0$-prevariety of stamps $\vst V$ is called a
\emph{$\mathfrak{C}_0$-variety of stamps} if it contains only
homomorphic images of Priestley stamps from~$\vst V$. In other words,
$\vst V$ is a $\mathfrak{C}_0$-variety of stamps if and only if
$\Phi(\Psi(\vst V))=\vst V$.

For the next lemma, recall the closure operator
$\__M$ 
on subsets of the Boolean algebra $\Cl B_{\mathrm{max}}^T$, whose
image consists of the M-closed subsets. Recall also that, whenever
$\Cl D$ is a $(0,1)$-sublattice of $\Cl B_{\mathrm{max}}^T$, its
M-closure $\Cl D_M$ is also a $(0,1)$-sublattice of $\Cl
B_{\mathrm{max}}^T$.

We call a $\mathfrak{C}_0$-prevariety of languages a
\emph{$\mathfrak{C}_0$-variety of languages} if, for every
$T\in\mathfrak{C}$, the $(0,1)$-sublattice $\Cl V(T)$ of $\Cl
B_{\mathrm{max}}^T$ is M-closed.

\begin{Lemma}
  \label{l:characterisation-of-stable-language-prevarieties} 
  Let  $\Cl V$  be a $\mathfrak{C}_0$-prevariety of languages, 
  and $T\in\mathfrak{C}$. Then $(\Psi(\Phi (\Cl V)))(T)=\Cl V(T)_M$.  
  In other words, $\Cl V = \Psi(\Phi (\Cl V))$ if and only if 
  $\Cl V$ is a $\mathfrak{C}_0$-variety of languages.  
\end{Lemma}
\begin{proof}
  For the whole proof we let $\vst V=\Phi (\Cl V)$.
  
  Let $L\in \Psi(\vst V)(T)$. We need to show that $L\in\Cl V(T) $.
  Since $\Psi(\vst V)(T)=\Rec_{\vst V}(T)$, the language $L$ is
  recognized by a stamp in $\vst V$. As $\vst V=\Phi (\Cl V)$, such a
  stamp is a homomorphic image of the stamp $\eta_{\Cl V(T)}$, and
  therefore $L$ is also recognized by the stamp $\eta_{\Cl V(T)}$.
  Recall that, for the $(0,1)$-sublattice $\Cl C=\Cl V(T)$ of $\Cl
  B_{\mathrm{max}}^T$, we have the Priestley stamp $\eta_{\Cl C} : T
  \to (\Stsynt(\Cl C),{\le_{\Cl C}})$
  and we know that $L\in \eta_{\Cl C}^{-1} (\Cl P_{uco} (\Stsynt (\Cl
  C)))$. By Corollary~\ref{c:St-C-vs-St-L} and the remark before the
  lemma, the set $\eta_{\Cl C}^{-1} (\Cl P_{uco} (\Stsynt (\Cl C)))$
  is in fact the sublattice $[\Cl C]=\Cl C_M=\Cl V(T)_M$. 
  
  Conversely, let $L\in \Cl V(T)_M$. By
  Corollary~\ref{c:St-C-vs-St-L}, we know that $L$ is recognized by
  $\eta_{\Cl V(T)}$. Since $\eta_{\Cl V(T)}\in \vst V=\Phi (\Cl V)$ we
  see that $L\in\Rec_{\vst V}(T)=\Psi(\vst V)(T)$.
\end{proof}

We may now establish the following Eilenberg type correspondence
theorem.

\begin{Thm}
  \label{t:Eilenberg}
  The restrictions of the mappings $\Psi$ and $\Phi$ are mutual
  inverse order-preserving bijections between
  $\mathfrak{C}_0$-varieties of stamps and $\mathfrak{C}_0$-varieties
  of languages.
\end{Thm}

\begin{proof}
  The statement is a straightforward consequence of the Galois
  connection in Proposition~\ref{p:Galois-connection} in view of
  Lemmas~\ref{l:characterisation-of-stable-stamps-prevarieties}
  and~\ref{l:characterisation-of-stable-language-prevarieties}.
\end{proof}

Now we concentrate on $\mathfrak{C}$-prevarieties of stamps, which are
$\mathfrak{C}_0$-prevarieties closed under the operator
$S_\mathfrak{C}$. The following lemma is an expected consequence of
Lemma~\ref{l:variety-of-stamps-recognition}.

\begin{Lemma}
  \label{l:recognition-functor-property}
  Let $\vst V$ be a $\mathfrak{C}$-prevariety of stamps and 
  $\psi : T \to T'$ be a morphism in the category $\mathfrak{C}$. Then
  \begin{displaymath}
    \psi^{-1}(\Rec_{\vst V}(T')) \subseteq 
    \Rec_{\vst V}(T).
  \end{displaymath}
\end{Lemma}

\begin{proof} 
  Let $L\in \Rec_{\vst V}(T')$. By
  Lemma~\ref{l:variety-of-stamps-recognition}, there is $\varphi : T'
  \to S$ in $\vst V$, where $(S,\le)$ is an ordered Stone topological
  algebra and a clopen upset $K$ in $S$ such that $L=\varphi^{-1}(K)$.
  Taking $\varphi\circ\psi : T \to S'$, where
  $S'=\overline{\mathrm{Im}(\varphi\circ\psi)}$, we get that
  $\psi^{-1}(L)=(\varphi\circ\psi)^{-1}(K\cap S')$, and $K\cap S'$ is
  a clopen upset in~$S'$. Since $\vst V$ is closed under the operator
  $S_{\mathfrak{C}}$ we conclude that $\psi^{-1}(L)\in \Rec_{\vst V}(T)$ by
  Lemma~\ref{l:variety-of-stamps-recognition}.
\end{proof}

By the previous lemma, we may extend the definition of 
$\Psi$ naturally, namely for every 
$\psi\in\mathfrak{C}(T,T')$, we put 
$\Psi(\vst{V}) (\psi)=\psi^{-1}|_{\Rec_{\vst V} (T')}$. 
In this way, for a $\mathfrak{C}$-prevariety $\vst V$
of stamps, the resulting $\Psi(\vst{V})$ is a
$\mathfrak{C}$-prevariety of languages.

For a $\mathfrak{C}$-prevariety of languages $\Cl V$, we have defined
$\Psi(\Cl V)$ as $H\Cl S$, where the class 
$\Cl S=\{\eta_{\Cl V(T)}; T\in \mathfrak{C}\}$ is formed by Priestley stamps. 
Note that, as there is at most one element of $\Cl S$ with a given
domain, $P\Cl S=\Cl S$, while $S_{\mathfrak{C}} \Cl S\subseteq H\Cl S$ 
by Lemma~\ref{l:recognition-functor-property} and Corollary~\ref{c:duality}.
Thus $\Psi(\Cl V)=H\Cl S=HPS_{\mathfrak{C}} \Cl S$
and it is a $\mathfrak{C}$-prevariety of stamps by Lemma~\ref{l:C-variety}. 

Altogether, the restriction of the mapping $\Psi$ to
$\mathbb{L}^{\mathfrak{C}}$ and the restriction of the mapping $\Phi$
to $\mathbb{S}^{\mathfrak{C}}$ are well defined correspondences. Thus,
we obtain from Proposition~\ref{p:Galois-connection} the following
special case.

\begin{Cor}
  \label{c:Galois-for-categories-C}
  The pair of mappings $(\Phi,\Psi)$ is a monotone Galois connection 
  between $\mathbb{L}^{\mathfrak{C}}$ and $\mathbb{S}^{\mathfrak{C}}$. 
\end{Cor}  
\begin{proof}
  Since the mappings are correctly defined, the statement follows from
  Proposition~\ref{p:Galois-connection}.
\end{proof}

And finally, we may formulate an Eilenberg type correspondence for
adequate subclasses. A $\mathfrak{C}$-variety of stamps is a
$\mathfrak{C}$-prevariety of stamps which is simultaneously a
$\mathfrak{C}_0$-variety of stamps. Similarly, a
$\mathfrak{C}$-variety of languages is a $\mathfrak{C}$-prevariety of
languages whose restriction to~$\mathfrak{C}_0$ is a
$\mathfrak{C}_0$-variety of languages.

\begin{Thm}
  \label{t:Eilenberg-version-C}
  The restrictions of the mappings $\Psi$ and $\Phi$ are mutual
  inverse order-preserving bijections between $\mathfrak{C}$-varieties
  of stamps and $\mathfrak{C}$-varieties of languages.
\end{Thm}

\begin{proof}
  The statement is an immediate consequence of
  Theorem~\ref{t:Eilenberg}.
\end{proof}

\section{A Birkhoff-type theorem for
  \texorpdfstring{$\mathfrak{C}$}{C}-varieties of stamps }
\label{sec:Reiterman}

Let $\mathfrak{C}$ be a category of topological $\Omega$-algebras.

For a topological algebra $T$, we say that $u\le v$ is a
\emph{pseudo-inequality over $T$} if $u,v\in(\Cl
B_{\mathrm{max}}^T)^\star$. In case $T$ is an object
from~$\mathfrak{C}$, we say that a $\mathfrak{C}$-stamp $\varphi:U\to
S$ \emph{satisfies} the pseudo-inequality $u\le v$ over $T$ and we may
write $\varphi\models u\le v$ if, for every morphism
$\psi\in\mathfrak{C}(T,U)$, we have $\delta(u)\le\delta(v)$ in $S$,
where $\delta$ is the unique continuous homomorphism $(\Cl
B_{\mathrm{max}}^T)^\star\to\overline{\varphi(\psi(T))}$ such that the
following diagram commutes:
\begin{equation}
  \label{eq:satisfaction}
  \begin{split}
  \xymatrix{
    T
    \ar[d]_{\iota_{\Cl B_{\mathrm{max}}^T}}
    \ar[rr]^\psi
    \ar[rd]^{\varphi\circ\psi}
    &&
    U
    \ar[d]^\varphi
    \\
    (\Cl B_{\mathrm{max}}^T)^\star
    \ar[r]^\delta
    &
    \overline{\varphi(\psi(T))}
    \ar@{^{(((}->}[r]
    &
    S
  }
  \end{split}
\end{equation}
In such a diagram, we call $\psi$ an \emph{evaluation} for the
pseudo-inequality $u\le v$. For a class $\Sigma$ of
pseudo-inequalities over objects of~$\mathfrak{C}$, we let
$\op\Sigma\cl_{\mathfrak{C}}$ be the class of all
$\mathfrak{C}$-stamps that satisfies all members of~$\Sigma$.

Let $\vst{V}$ be a $\mathfrak{C}$-variety of stamps and $T$ an object
from~$\mathfrak{C}$. We denote by $\xi_{T,\vst{V}}:T\to F_T\vst{V}$
the product of all stamps $T\to S$ from~$\vst{V}$. Thus,
$\xi_{T,\vst{V}}$ is the most general stamp from $\vst{V}$ with domain
$T$, in the sense that all other are morphic images of it.

\begin{Lemma}
  \label{l:Birkhoff}
  Let $T$ be an object of~$\mathfrak{C}$ and let $\Cl B=\Cl
  B_{\mathrm{max}}^T$. Let $\zeta:\Cl B^\star\to F_T\vst{V}$ be the
  unique continuous homomorphism such that $\zeta\circ\iota_{\Cl
    B}=\xi_{T,\vst{V}}$. Then $\vst{V}$ satisfies a pseudo-inequality
  $u\le v$ over $T$ if and only if $\zeta(u)\le\zeta(v)$.
\end{Lemma}

\begin{proof}
  Assume first that $\vst{V}$ satisfies the pseudo-inequality $u\le v$
  over $T$. Since $\xi_{T,\vst{V}}$ is a stamp from~$\vst{V}$ the
  equality $\zeta(u)\le\zeta(v)$ is a special case of the verification
  that $\xi_{T,\vst{V}}$ satisfies $u\le v$, namely taking the
  evaluation $\psi=\mathrm{id}_T$.

  Conversely, suppose that $\zeta(u)\le\zeta(v)$ and consider an
  arbitrary stamp $\varphi:U\to S$ from~$\vst{V}$ and a
  $\mathfrak{C}$-morphism $\psi:T\to U$. Then, there exists an
  order-preserving continuous homomorphism $\gamma$ such that the
  following diagram commutes
  \begin{displaymath}
    \xymatrix{
      &
      T
      \ar[dl]_{\iota_{\Cl B}}
      \ar[d]^{\xi_{T,\vst{V}}}
      \ar[rr]^\psi
      \ar[rd]^{\varphi\circ\psi}
      &&
      U
      \ar[d]^\varphi
      \\
      \Cl B^\star
      \ar[r]^\zeta
      &
      F_T\vst{V}
      \ar[r]^(.4)\gamma
      &
      \overline{\varphi(\psi(T))}
      \ar@{^{(((}->}[r]
      &
      S
    }
  \end{displaymath}
  Thus, the continuous homomorphism $\delta$ of
  Diagram~(\ref{eq:satisfaction}) is $\gamma\circ\zeta$. Since
  $\gamma$ preserves order and $\zeta(u)\le\zeta(v)$ in~$F_T\vst{V}$,
  we deduce that $\delta(u)\le\delta(v)$ in $S$, thereby showing that
  the stamp $\varphi$ satisfies $u\le v$.
\end{proof}

The following is the analog for $\mathfrak{C}$-varieties of stamps of
Birkhoff's variety theorem.

\begin{Thm}
  \label{t:Birkhoff}
  Let $\vst{V}$ be a class of $\mathfrak{C}$-stamps. Then $\vst{V}$ is
  a $\mathfrak{C}$-variety if and only if it is of the form
  $\op\Sigma\cl_{\mathfrak{C}}$ for some class $\Sigma$ of
  pseudo-inequalities over objects of~$\mathfrak{C}$.
\end{Thm}

\begin{proof}
  It is routine to check that, if $\Sigma$ is a class of
  pseudo-inequalities over objects of~$\mathfrak{C}$ then
  $\op\Sigma\cl_{\mathfrak{C}}$ is a $\mathfrak{C}$-variety of stamps.
  Suppose next that $\vst{V}$ is a $\mathfrak{C}$-variety of stamps.
  Consider the class $\Sigma$ of all pseudo-inequalities over objects
  of~$\mathfrak{C}$ that are satisfied by all stamps from~$\vst{V}$
  and the $\mathfrak{C}$-variety of stamps
  $\vst{W}=\op\Sigma\cl_{\mathfrak{C}}$ it defines. As $\vst{V}$ is
  contained in~$\vst{W}$ by definition of the latter, to conclude the
  proof it suffices to establish the claim that every
  $\mathfrak{C}$-stamp $\varphi:T\to S$ from $\vst{W}$ belongs
  to~$\vst{V}$.

  To simplify the notation, let $\Cl B=\Cl B_{\mathrm{max}}^T$. Let
  $\delta:\Cl B^\star\to S$ be the unique continuous homomorphism such
  that $\delta\circ{\iota_{\Cl B}}=\varphi$.
  We claim that, for $u,v\in \Cl B^\star$, $\zeta(u)\le\zeta(v)$ in
  $F_T\vst{V}$ implies $\delta(u)\le\delta(v)$ in~$S$. Indeed, by
  Lemma~\ref{l:Birkhoff}, the assumption $\zeta(u)\le\zeta(v)$ entails
  that the pseudo-inequality $u\le v$ holds in~$\vst{V}$, so that it
  belongs to~$\Sigma$ and, therefore it is satisfied by $\varphi$,
  whence $\delta(u)\le\delta(v)$ in~$S$. Since both orders on
  $F_T\vst{V}$ and $S$ are stable partial orders, we conclude that
  there is an order-preserving continuous homomorphism
  $\varepsilon_{T,\vst{V}}$ such that the following diagram commutes:
  \begin{displaymath}
    \xymatrix{
      &
      T
      \ar@/_5mm/[ldd]_{\xi_{T,\vst{V}}}
      \ar[d]^{\iota_{\Cl B}}
      \ar@/^5mm/[rdd]^\varphi
      &
      \\
      &
      \Cl B^\star
      \ar[ld]_\zeta
      \ar[rd]^\delta
      &
      \\
      F_T\vst{V}
      \ar[rr]_{\varepsilon_{T,\vst{V}}}
      &&
      S   
    }
  \end{displaymath}
  Hence, $\varphi$ is a homomorphic image of~$\xi_{T,\vst{V}}$ and, therefore,
  $\varphi$ belongs to~$\vst{V}$.
\end{proof}

Note that, if $\mathfrak{C}$ is a small category, then the class of
all pseudo-inequalities over objects from~$\mathfrak{C}$ is actually a set.
It follows that then the class $\Sigma$ of the proof of
Theorem~\ref{t:Birkhoff} is also a set.

Among pseudo-inequalities those of the form $\iota_{\Cl
  B^T_{\mathrm{max}}}(s)\le\iota_{\Cl B^T_{\mathrm{max}}}(t)$ with
$s,t\in T$ are particularly simple. The next result characterizes the
$L\in\Cl B_\mathrm{max}^T$ such that $\Stsynt(L)$ satisfies
$\iota_{\Cl B^T_{\mathrm{max}}}(s)\le\iota_{\Cl
  B^T_{\mathrm{max}}}(t)$.

\begin{Prop}
  \label{p:syntactic-inequalities}
  Let $L$ be an admissible subset of\/~$T$, where $T$ is an object
  from~$\mathfrak{C}$. Then, for $u,v\in U$, with $U$ also an object
  in~$\mathfrak{C}$, the pseudo-inequality $u\le v$ holds in the
  $\mathfrak{C}$-stamp $\eta_L:T\to\Stsynt(L)$ if and only if
  $\psi(u)\preccurlyeq_L\psi(v)$ in~$T$ for every morphism
\end{Prop}

\begin{proof}
  In view of Proposition~\ref{p:syntactic-in-dual} and the definition
  of satisfaction of a pseudo-inequality, it suffices to observe that
  the pseudo-inequality $u\le v$ holds in the $\mathfrak{C}$-stamp
  $\eta_L:T\to\Stsynt(L)$ if and only if
  $\eta_L(\psi(u))\le\eta_L(\psi(v))$.
\end{proof}

\section{Minimal automata}
\label{sec:min-automata}

This section concerns languages over a fixed finite alphabet $A$, in
the sense of subsets of the free monoid $A^*$. As Stone topological
semigroups are profinite, they only recognize regular languages by
clopen sets. An alternative is to consider $A^*$ as a unary algebra
under the action of the letters by right multiplication, which is in
fact the free monogenic algebra in the signature $\Omega=\Omega_1=A$.

For a topological $\Omega$-algebra $T$, the linear transformations of
$T$ are then functions of the form $f:t\mapsto tw$ with $w\in A^*$.
Thus, when $T=A^*$, for a language $L\subseteq A^*$, we have
$f^{-1}(L)=Lw^{-1}$. Hence, $[L]=[\Cl C_L]$ is the $(0,1)$-sublattice
of~$\Cl P(A^*)$ generated by the languages of the form $Lw^{-1}$ with
$w\in A^*$. By Proposition~\ref{p:syntactic-in-dual}, the image of the
Priestley completion $\eta_L:A^*\to\Stsynt(L)$ is precisely the
syntactic $\Omega$-algebra $\synt_\Omega(L)$ of~$L$ over~$A^*$. Note
that, for $u,v\in A^*$ we have the following equivalences:
\begin{align*}
  u\curlyeqprec_Lv
  &\iff
    \forall w\in A^*\ (u\in Lw^{-1} \implies v\in Lw^{-1})\\
  &\iff
    \forall w\in A^*\ (uw\in L \implies vw\in L)\\
  &\iff
    \forall w\in A^*\ (w\in u^{-1}L \implies w\in v^{-1}L)\\
  &\iff
    u^{-1}L\subseteq v^{-1}L.
\end{align*}
Thus, $\synt_\Omega(L)$ is the familiar minimal semiautomaton of~$L$,
whose states are the languages of the form $u^{-1}L$ and the action of
each letter $a\in A$ is given by $(u^{-1}L)\cdot
a=a^{-1}(u^{-1}L)=(ua)^{-1}L$. The \emph{minimal automaton} of $L$ is
obtained by adding the requirement that $L=1^{-1}L$ is the initial
state and that a state $u^{-1}L$ is final if and only if $u\in L$,
noting that, by the above equivalences, if $u^{-1}L=v^{-1}L$ then
$u\in L$ if and only if $v\in L$.

Steinberg \cite{Steinberg:2013} proposed the following
compactification of the minimal automaton $\synt_\Omega(L)$. He first
considers the set $\Cl P(A^*)$ of all languages $K$ over $A$ as the
Cantor space $X=\{0,1\}^{A^*}$ of their characteristic functions
$\chi_K$ (defined by $\chi_K^{-1}(1)=K$). Note that the function
$a_X:X\to X$ sending $\chi_K$ to $\chi_K(a\_)=\chi_{a^{-1}K}$ is
continuous. Thus, $X$ is a Stone topological $\Omega$-algebra. Within
it, the subalgebra generated by $\chi_L$ is precisely the minimal
semiautomaton $\synt_\Omega(L)$. Steinberg considers within the space
$X$ the closure of the universe of the algebra $\synt_\Omega(L)$,
which is $\chi_L\cdot A^*$, which becomes a Stone topological
$\Omega$-subalgebra $X(L)$ of~$X$.

\begin{Prop}
  \label{p:Steinberg-minimal-automaton}
  The mapping $\varphi:(\chi_L\cdot A^*)\to\synt(A^*)$ sending $\chi_{w^{-1}L}$
  to $w^{-1}L$ extends to an isomorphism of Stone topological algebras
  $X(L)\to\Stsynt(L)$.
\end{Prop}

\begin{proof}
  Consider the two Stone completions $\psi:A^*\to X(L)$, given by
  $\psi(w)=\chi_{w^{-1}L}$, and $\eta_L:A^*\to\Stsynt(L)$. Note that
  $\varphi$ is such that $\varphi\circ\psi=\eta_L$. By the equivalence
  theorem \cite[Theorem~7.2]{Almeida&Klima:2024a} these two Stone
  completions are isomorphic (through an extension of $\varphi$) if and
  only if the Boolean algebras $\Cl B_1$ and $\Cl B_2$ of preimages
  under $\psi$ or $\eta_L$ of respectively clopen subsets of $X(L)$ or
  $\Stsynt(L)$ coincide. From basic duality theory
  (cf.~\cite[Proposition~5.1]{Almeida&Klima:2024a}) and
  Corollary~\ref{c:B(L)}, we know that $\Cl B_2$ is precisely the
  Boolean subalgebra of $\Cl P(A^*)$ generated by the languages of the
  form $Lw^{-1}$ with $w\in A^*$.

  On the other hand, the clopen subsets of $X(L)$ are the
  intersections with $X(L)$ of clopen subsets of~$X$; the latter are
  Boolean combinations of sets (cylinders) of the form
  \begin{displaymath}
    C_{F,G}=\{f\in X: f(F)=\{1\}\wedge f(G)=\{0\}\}
  \end{displaymath}
  where $F$ and $G$ are finite languages. Now, for $w\in A^*$, we have
  \begin{align*} \chi_{w^{-1}L}\in C_{F,G} &\iff F\subseteq w^{-1}L
    \wedge G\cap w^{-1}L=\emptyset\\ &\iff wF\subseteq L \wedge wG\cap
    L=\emptyset\\ &\iff w\in\bigcap_{u\in F}Lu^{-1} \wedge
    w\notin\bigcup_{v\in G}Lv^{-1}\\ &\iff w\in \bigcap_{u\in
      F}Lu^{-1} \setminus\bigcup_{v\in G}Lv^{-1}
  \end{align*} which shows that $\psi^{-1}(C_{F,G})\in\langle
  L\rangle$, so that $\Cl B_1\subseteq\langle L\rangle$. The reverse
  inclusion holds since $L=\psi^{-1}(C_{\{1\},\emptyset})\in\Cl B_1$ and
  $\Cl B_1$ is M-closed by Proposition~\ref{p:M-closed}. Hence, we have
  $\Cl B_1=\Cl B_2$, which establishes the proposition.
\end{proof}

An alternative proof of
Proposition~\ref{p:Steinberg-minimal-automaton} is obtained by
combining the equivalence theorem
\cite[Theorem~7.2]{Almeida&Klima:2024a} with
\cite[Theorem~2.2]{Steinberg:2013}, which basically avoids the
calculations in the second paragraph of the above proof. The above
proof seems preferable as it does not depend on results that so far
have not been presented beyond the preprint form.

Note that, on the free monoid $A^*$, viewed as an algebra in the
signature $\Sigma=\Sigma_0\cup\Sigma_2$ with $\Sigma_0=\{1\}$ and
$\Sigma_2=\{{\cdot}\}$, for a language $L\subseteq A^*$, the algebraic
syntactic quasi-order $\curlyeqprec_L$ considered
in~(\ref{eq:algebraic-syntactic-qo}) is given by
\begin{equation}
  \label{eq:algebraic-syntactic-qo-binary}
  u \curlyeqprec_L^\Sigma v
  \iff
  \forall x,y\in A^*\ (xuy\in L \implies xvy\in L),
\end{equation}
which is the definition of syntactic order originally considered by
Sch\"utzenberger \cite[page 15-10]{Schutzenberger:1956b} although it
was later extensively used in the French school in the reverse form
\cite[Section 3.2]{Pin:1997}.

For an alphabet $A$, let $\mathfrak{C}_A$ be the category with the
monogenic free unary $A$-algebra $A^*$ as single object and its
endomorphisms as morphisms. Proposition~\ref{p:syntactic-inequalities}
yields the following interpretation of validity of inequalities in the
completed minimal automaton $\Stsynt^\Omega(L)$ of a language
$L\subseteq A^*$.

\begin{Prop}
  \label{p:syntactic-inequalities-unary-vs-binary}
  Let $L$ be an arbitrary language over a finite alphabet $A$. Then,
  for $u,v\in A^*$, the pseudo-inequality $u\le v$ holds in the
  $\mathfrak{C}_A$-stamp $\eta_L:A^*\to\Stsynt(L)$ if and only if
  $u\curlyeqprec_L^\Sigma v$.
\end{Prop}

\begin{proof}
  Since the endomorphisms of $A^*$ as a unary $A$-algebra are given by
  multiplication on the left by the same word $x$,
  Proposition~\ref{p:syntactic-inequalities} yields the first of the
  following equivalences:
  \begin{align*}
    \eta_L^\Omega \models u\le v
    &\iff
    \forall x\in A^*\ (xu \le_L^\Omega xv) \\
    &\iff
    \forall x,y\in A^*\ (xu\in Ly^{-1} \implies xv\in Ly^{-1}) \\
    &\iff
    \forall x,y\in A^*\ (xuy\in L \implies xvy\in L) \\
    &\iff
    u \curlyeqprec_L^\Sigma v,
  \end{align*}
  where the second and last equivalences are given by the definitions
  of the partial order~$\le_L^\Omega$
  (see~\ref{eq:qo-for-M-closed-admissible} and
  Corollary~\ref{c:admissible-to-po}) and
  quasi-order~$\curlyeqprec_L^\Sigma$
  (\ref{eq:algebraic-syntactic-qo-binary}).
\end{proof}

We may in fact characterize membership in any $\mathfrak{C}_A$-variety
of languages in terms of words as follows.

\begin{Prop}
  \label{p:inequality-definability}
  Let $A$ be an alphabet and let $\Cl C$ be a set of languages
  over~$A$. Then a language $L\subseteq A^*$ is not in the
  $\mathfrak{C}_A$-variety of languages $\Cl V(\Cl C)$ generated
  by~$\Cl C$ if and only if there is a net $(u_i,v_i)_{i\in I}$ in
  $A^*\times A^*$ such that the following condition holds for all
  $K=x^{-1}Jy^{-1}$ with $J\in\Cl C$ and $x,y\in A^*$ but not for
  $K=L$:
  \begin{equation}
    \label{eq:inequality}
    \exists i_0\,\forall i\ge i_0\ (u_i\in K\implies v_i\in K).
  \end{equation}
\end{Prop}

\begin{proof}
  Since $\Cl V(\Cl C)$ is a $\mathfrak{C}_A$-variety of languages, in
  view of Theorem~\ref{t:Birkhoff} its corresponding
  $\mathfrak{C}_A$-variety of stamps (according to
  Theorem~\ref{t:Eilenberg-version-C}) is defined by the set $\Sigma$
  of all pseudo-inequalities that are valid in all
  $\mathfrak{C}_A$-stamps $\eta_K:A^*\to\Stsynt(K)$ with $K\in\Cl C$.
  Hence $L$ does not belong to~$\Cl V(\Cl C)$ if and only if there is
  a pseudo-inequality $s\le t$ from $\Sigma$ that fails in~$\eta_L$.

  It remains to express in terms of word membership the validity of
  pseudo-inequalities. For that purpose, let
  $\bigl((s_i,t_i)\bigr)_{i\in I}$ be a net in $A^*\times A^*$
  converging to $(s,t)$ in the product space
  $\beta(A^*)\times\beta(A^*)$. Then, for a language $K\subseteq A^*$,
  we have the following equivalences and implication:
  \begin{align}
    \lefteqn{\eta_K\models s\le t}
    \nonumber\\
    &\iff
      \forall x,y\in A^*\ (x^{-1}Ky^{-1}\in s \implies x^{-1}Ky^{-1}\in t)
      \nonumber
    \\
    &\iff \nonumber\\
     \lefteqn{\quad\ \forall x,y\in A^*\ \big(
         \exists i_0\,\forall i\ge i_0\ (s_i\in x^{-1}Ky^{-1})}
     \label{eq:inequality-2}\\
    &\hphantom{\quad\ y\in A^*\ \ \implies}
      \implies
         \exists j_0\,\forall i\ge j_0\ (t_i\in x^{-1}Ky^{-1})
         \big)
      \nonumber\\
    &\iff \nonumber\\
    \lefteqn{\quad\ \forall x,y\in A^*\
    \exists k_0\,\forall i\ge k_0\ 
    (s_i\in x^{-1}Ky^{-1} \implies t_i\in x^{-1}Ky^{-1})}
    \label{eq:inequality-3}
  \end{align}
  where the last equivalence is the only one requiring some
  justification.

  The implication
  $(\ref{eq:inequality-3})\Rightarrow(\ref{eq:inequality-2})$
  is obtained by taking $j_0$ to be a common upper bound of $k_0$ and
  $i_0$ if the latter exists.

  For the implication
  $(\ref{eq:inequality-2})\Rightarrow(\ref{eq:inequality-3})$,
  suppose that (\ref{eq:inequality-2}) holds and let $x,y\in A^*$.
  We distinguish two cases.

  \smallskip

  \textbf{Case 1:} Suppose that there is $i_0$ such that, for all
  $i\ge i_0$, we have $s_i\in x^{-1}Ky^{-1}$. Then we may choose $j_0$
  as in~(\ref{eq:inequality-2}) so that any upper bound $k_0$ of
  $i_0$ and $j_0$ is such that, for all $i\ge k_0$, $s_i\in
  x^{-1}Ky^{-1}$ implies $t_i\in x^{-1}Ky^{-1}$.

  \smallskip
  
  \textbf{Case 2:} Case 1 does not hold. Then, there is a subnet
  $(s_{i_\lambda})_\lambda$ in the complement $A^*\setminus
  x^{-1}Ky^{-1}$. As such a subnet still converges to $s$, we conclude
  that $A^*\setminus x^{-1}Ky^{-1}\in s$. But, then, as $s=\lim s_i$,
  there must exist $i_0$ such that, for all $i\ge i_0$, we have
  $s_i\notin x^{-1}Ky^{-1}$ Hence, we may choose $k_0=i_0$ to get, for
  all $i\ge k_0$, $s_i\in x^{-1}Ky^{-1}$ implies $t_i\in
  x^{-1}Ky^{-1}$ trivially.

  \smallskip

  Thus, since, by the assumption that $s\le t$ fails in $\eta_L$,
  there exist $x,y\in A^*$ such that (\ref{eq:inequality}) fails when
  we take $u_i=xs_iy$ and $v_i=xt_iy$ ($i\in I$). On the other hand,
  the equivalence of~(\ref{eq:inequality-2}) with $\eta_K\models s\le
  t$, shows that the pseudo-inequality $xsy\le xty$ remains valid,
  where ``left multiplication by $x$'' in~$\beta(A^*)$ is the unique
  extension of left multiplication by~$x$ in~$A^*$ to a continuous
  transformation of~$\beta(A^*)$. Hence, (\ref{eq:inequality}) holds
  whenever $K\in\Cl V(\Cl C)$. To complete the proof of the theorem,
  it remains to observe that Condition~(\ref{eq:inequality}) is
  preserved under taking finite unions and intersections.
\end{proof}

We conclude this section by improving
Proposition~\ref{p:inequality-definability} by replacing nets by
sequences in case the alphabet is finite and we start with a countable
class of languages. Both assumptions are innocuous as classes of
languages on an alphabet of interest in applications assume that the
alphabet is finite and that the languages or their complements are at
least recursively enumerable. The fact that we fix the alphabet may
seem too restrictive but, for the purpose of separation, it is clearly
sufficient.

\begin{Thm}
  \label{t:inequality-definability-sequences}
  Let $A$ be a finite alphabet and let $\Cl C$ be a countable set of
  languages over~$A$. Then a language $L\subseteq A^*$ is not in the
  $\mathfrak{C}_A$-variety of languages $\Cl V(\Cl C)$ generated
  by~$\Cl C$ if and only if there are sequences
  $(s_n)_{n\in\mathbb{N}}$ in~$L$ and $(t_n)_{n\in\mathbb{N}}$ in
  $A^*\setminus L$ such that the following condition holds for all
  $K=x^{-1}Jy^{-1}$ with $J\in\Cl C$ and $x,y\in A^*$:
  \begin{equation}
    \label{eq:inequality-sequences}
    \exists n_0\,\forall n\ge n_0\ (s_n\in K\implies t_n\in K).
  \end{equation}
\end{Thm}

\begin{proof}
  Suppose that $L\notin\Cl V(\Cl C)$ and let
  $\bigl((u_i,v_i)\bigr)_{i\in I}$ be a net given by
  Proposition~\ref{p:inequality-definability}. Let $K_0,K_1,\ldots$ be
  an enumeration of the languages of the form $x^{-1}Jy^{-1}$ with
  $J\in\Cl C$ and, for each $k\in\mathbb{N}$, let $j_k$ be a choice of
  witness $i_0$ for Condition~(\ref{eq:inequality}) for the language
  $K=K_n$.

  We define a function $\varphi:\mathbb{N}\to I$ recursively as
  follows. By the choice of net of pairs of words, we know that
  \begin{equation*}
    \label{eq:inequality-sequences-2}
    \forall i\,\exists j\ge i\ (u_j\in L \wedge v_j\notin L).
  \end{equation*}
  We may thus choose $\varphi(0)\in I$ such that $\varphi(0)\ge j_0$,
  $u_{\varphi(0)}\in L$, and $v_{\varphi(0)}\notin L$. Assuming that
  $\varphi(k)$ ($k=0,\ldots,n-1$) have already been chosen in
  increasing order such that $\varphi(k)\ge j_k$, $u_{\varphi(k)}\in
  L$ and $v_{\varphi(k)}\notin L$, let $\varphi(n)$ be any element
  of~$I$ which is larger than $\varphi(n-1)$ and $j_n$ such that
  $u_{\varphi(n)}\in L$ and $v_{\varphi(n)}\notin L$.

  We claim that the choice of sequences $s_n=u_{\varphi(n)}$ and
  $t_n=v_{\varphi(n)}$ satisfies (\ref{eq:inequality-sequences}) for
  every language $K=x^{-1}Jy^{-1}$ with $J\in\Cl C$, and $x,y\in A^*$.
  Indeed, there is $n\in\mathbb{N}$ such that $K=K_n$ and the claim
  follows from~(\ref{eq:inequality}) since we have chosen
  $\varphi(n)\ge j_n$ and the sequence $(\varphi(n))_n$ is increasing.
\end{proof}

\section{Bounding context-free languages}
\label{sec:Pumping-CF}

In this section, we consider the class $\mathbf{CF}(A)$ of
context-free languages over a finite alphabet $A$ and the
$\mathfrak{C}_A$-variety of languages $\mathbf{ICF}(A)$ it generates,
where the category $\mathfrak{C}_A$ was introduced in the preceding
section. It is easy to see that $\mathbf{ICF}(A)$ consists of all
finite intersections of languages from $\mathbf{CF}(A)$, and it is a
class of languages about which not much seems to be known; it is
interesting as for all such languages the complexity of the membership
problem is the same as for context-free languages and so it is
sub-cubic \cite{Valiant:1975,Williams&Xu&Xu&Zhou:2024}. We illustrate
how the general separation criterion given by
Theorem~\ref{t:inequality-definability-sequences} can be used to rule
out that languages belong to $\mathbf{ICF}(A)$.

The key ingredient is to obtain suitable sequences of words in~$A^*$
such that Condition~(\ref{eq:inequality-sequences}) holds for every
language $K\in\mathbf{CF}(A)$. It may be thought as a kind of
replacement or pumping property. We did not find our first example in
the literature, even though many pumping properties of regular and
context-free languages are known. The classical pumping lemma for
context-free languages is the following statement which is mainly used
to show that a concrete given language is not context-free.

\begin{Lemma}[{\cite[Theorem 7.18]{Hopcroft&Motwani&Ullman:2001}}]
  \label{l:pumping}
  For every context-free language $L$, there exists a positive integer
  $n$ satisfying the following property: if $z\in L$ is a word of
  length at least $n$, then $z$ can be factorized as $z=uvwxy$, where
  \begin{enumerate}
  \item the length of the word $vwx$ is at most $n$,
  \item the word $vx$ is not empty, and
  \item for each $i\in \mathbb N$, the word $uv^iwx^iy$ belongs
    to~$L$.
  \end{enumerate}
\end{Lemma} 

Informally, our new pumping lemma is trying to express the property of
context-free languages that whenever we find a pattern of the form
$v^iwx^i$ in a word of the language for a large enough number $i$,
then the powers $v^i$ and $x^i$ must come from (synchronized) pumping,
and thus we may continue to pump the corresponding (powers of the)
factors $v$ and $x$. A formal formulation of this property is at first
made in the special case where $v$, $w$, and $x$ are letters and $u$
and $y$ are the empty word. Then this special case is extended to the
general one. Although we contribute to the theory of context-free
languages, we do not give here an introduction to that theory; we only
apply some well-known results, such as the closure properties of the
class of all context-free languages. For missing details and notions
we refer to~\cite{Autebert&Berstel&Boasson:1997} and \cite[Chapters 5
and 7]{Hopcroft&Motwani&Ullman:2001}.

\begin{Lemma}
  \label{l:pumping-a-and-c-to-n-factorial}
  Let $L$ be a context-free language over the alphabet $A$ containing
  three different letters $a,b,c$. Then there exists an integer $n_0$
  such that, for every $n \ge n_0$, the following property holds:
  \begin{displaymath}
    a^{n!}bc^{n!}\in L
    \implies
    \exists k < n_0\, \forall p\in\mathbb N \
    (a^{n!+pk}bc^{n!+pk}\in L).
  \end{displaymath}
\end{Lemma}  

\begin{proof}
  We plan to apply Parikh's
  theorem~\cite[Theorem~2.6]{Autebert&Berstel&Boasson:1997} which
  characterizes the images of context-free languages under the
  homomorphism $\varphi : A^* \rightarrow \mathbb N^A$ that maps words
  to their Parikh's vector, that is, the sequence $(|u|_e)_{e\in A}$
  where $|u|_e\ge 0$ denotes the number of occurrences of the letter
  $e$ in the word $u$.
  Parikh's theorem states, for a context-free language $K$, that the
  image $\varphi(K)$ of the language $K$ in the considered morphism is
  a semi-linear subset in $\mathbb{N}^A$. Recall that a subset of
  $\mathbb N^A$ is \emph{linear}, if it can be expressed as $\{ u +
  \Sigma_{i=1}^{\ell} t_i v_i : t_i\in \mathbb N\}$ where
  $\ell\in\mathbb{N}$ and $u$ as well as $v_1,\dots , v_\ell$ are
  fixed members of $\mathbb N^A$. A subset of $\mathbb N^A$ is
  \emph{semi-linear} if it may be expressed as a finite union of
  linear subsets. For our purpose, we apply Parikh's theorem to the
  language $L\cap a^*bc^*$, which is context-free because so is every
  intersection of a context-free language with a regular one. Notice
  that the restriction of $\varphi$ to the domain $L\cap a^*bc^*$ is
  an injective mapping, and we may assume that $A$ consists exactly of
  the three letters $a,b,c$. For the usual order of the letters, we
  may identify $\mathbb{N}^A$ with $\mathbb{N}^3$. Thus, if we denote
  $K=\varphi(L\cap a^*bc^*) \subseteq \mathbb N^3$, the
  statement of our lemma may be reformulated in this way: there exists
  an integer $n_0$ such that, for every $n \ge n_0$, the following
  property holds:
  \begin{displaymath}
    (n!,1,n!)\in K
    \implies
    \exists k < n_0\, \forall p\in \mathbb N\ 
    \bigl((n!+pk,1,n!+pk)\in K\bigr).
  \end{displaymath}
  The rest of the proof is devoted to showing that statement.
  
  Let $m$ be the maximum of the coordinates of the vectors in a fixed
  expression of the semi-linear subset $K$ as a finite
  union of linear subsets of $\mathbb{N}^3$. Put $n_0=(m+1)^2$.
  Let $n\ge n_0$ be arbitrary and assume that $(n!,1,n!)$ belongs to
  $S=\{ u + \Sigma_{i=1}^{\ell} t_i v_i : t_i\in \mathbb N\}\subseteq
  K$, where $S$ is an appropriate linear subset of $\mathbb N^3$ from
  the considered expression of $K$. Let $\overline{t_i}$
  ($i=1,\ldots,\ell$) be an appropriate choice of coefficients
  witnessing $(n!,1,n!)\in S$, that is, coefficients such that
  \begin{equation}
    \label{eq:n-factorial-in-S}
    u + \Sigma_{i=1}^{\ell} \overline{t}_i v_i=(n!,1,n!).
  \end{equation}
  We know that $\ell \ge 1$ because $u$ has all coordinates smaller
  than $n$, so there is at least one non-zero $\overline{t}_i$. Since
  all members of $K$ have the second coordinate equal to $1$, the
  second coordinate of $u$ is $1$ and the second coordinate of all
  $v_i$ is $0$. In the rest of the proof, we distinguish four cases
  according to properties of the vectors $v_i$. The assumptions of
  these cases are formulated in such a way that, evidently, at least
  one of the cases has to hold. Nevertheless, we do not state that the
  cases are mutually exclusive; in fact, we anticipate that only the
  first two cases may hold, even though, they may both hold. Let $p\in
  \mathbb N$ be arbitrary.

  \smallskip
  
  {\bf Case 1:} Assume that there is an index $i$ such that 
  $v_i=(r,0,r)$ for some $0<r\le m$. 
  Then we put $k=r$ for which we have $k\le m<n_0$.  
  Furthermore, we put $t_j=\overline{t}_j$ for all 
  $j\in\{1,\dots, i-1, i+1,\dots,\ell\}$ and $t_i=\overline{t}_i+p$. 
  In this way, we obtain the vector
  $s=(n!,1,n!)+p(k,0,k)$ from the set $S$, which shows that
  $s=(n!+pk,1,n!+pk)\in K$.
  
  \smallskip
  
  {\bf Case 2:} Assume that there is an index $i$ such that
  $v_i=(q,0,r)$, where $q>r$, and there is an index $i'$ such that
  $v_{i'}=(q',0,r')$, where $q'<r'$. Now we put $k=qr'-q'r$ for which
  we have $0<k\le qr'\le m^2<n_0$. Furthermore, we put
  $t_i=\overline{t}_i+p (r'-q')$,
  $t_{i'}=\overline{t}_{i'}+p (q-r)$, and $t_j=\overline{t}_j$
  for all the other indices $j$. In this way, we obtain the following
  vector $s$ in $S$:
  \begin{displaymath}
    s=(n!,1,n!)+p (r'-q') (q,0,r) + p (q-r) (q',0,r').
  \end{displaymath}
  In the first coordinate of $s$, we add to $n!$ the $p$-multiple of
  \begin{displaymath}
    (r'-q') q +(q-r) q' = qr'-rq'=k.
  \end{displaymath}
  Similarly, in the third coordinate we add to $n!$ the $p$-multiple
  of
  \begin{displaymath}
    (r'-q') r +(q-r) r' = qr'-rq'=k .
  \end{displaymath}
  Thus $s=(n!,1,n!)+p(k,0,k)=(n!+pk,1,n!+pk)$ which belongs to
  $K$ as we wanted to show.
  
  \smallskip
  
  {\bf Case 3:}  Assume that for all vectors $v_i=(q_i,0,r_i)$, we have $q_i>r_i$.
  We compute the difference between the first and the third coordinate in the 
  vector $(n!,1,n!)$ using expression~(\ref{eq:n-factorial-in-S}). 
  For that purpose, we denote $u=(q,1,r)$. We may deduce that $q<r$, and the 
  difference $r-q$ is at most $m$.
  Since the difference between the first and third coordinate 
  in~(\ref{eq:n-factorial-in-S}) is $0$, we get the equality 
  $r-q=\Sigma_{i=1}^{\ell} \overline{t}_i (q_i-r_i)$. 
  For every index $i$, we have 
  $q_i-r_i\ge 1$, thus we get $m\ge r-q \ge \Sigma_{i=1}^{\ell} \overline{t}_i$.
  Using this inequality, we may also deduce, 
  for the first coordinate in~(\ref{eq:n-factorial-in-S}), 
  the following sequence of inequalities:
  \begin{displaymath}
    n!
    =   q + \Sigma_{i=1}^{\ell} \overline{t}_i q_i
    \le m + \Sigma_{i=1}^{\ell} \overline{t}_i m
    \le m+m^2
    <(m+1)^2
    =n_0
    \le n .
  \end{displaymath}
  This is a contradiction because the inequalities $n! <(m+1)^2\le n$
  are impossible. Thus, case 3 does not occur.
  
  \smallskip
  
  {\bf Case 4:} If for all vectors $v_i=(q_i,0,r_i)$, we have $q_i<r_i$, 
  then we may proceed in the same way as in case 3. 

  \smallskip

  Altogether, since cases 3 and 4 do not hold, at least one of cases 1
  and 2 occurs. In both of them, we found an integer $k<n_0$ with the
  required property.
\end{proof}

We are ready to formulate and prove the general case of our cumulative 
pumping property. 

\begin{Prop}
  \label{p:pumping-a-to-n-factorial-b-c}
  Let $A$ be an alphabet, $L$ be a context-free language over~$A$, and
  $u,v,w,x,y\in A^*$ be words. Then there exists an integer $n_0$ such
  that, for every $n \ge n_0$, the following property holds:
  \begin{displaymath}
    uv^{n!}wx^{n!}y\in L
    \implies
    \exists k < n_0\, \forall p\in\mathbb N\ 
    (uv^{n!+pk}wx^{n!+pk}y\in L) .
  \end{displaymath}
  In particular, we have
  \begin{equation}
    \label{eq:pumping-a-to-n-factorial-b-c}
    \exists n_0\,\forall n\ge n_0\
    (uv^{n!}wx^{n!}y\in L \implies uv^{2(n!)}wx^{2(n!)}y\in L).
  \end{equation}
\end{Prop} 

\begin{proof}
  We consider the language $u^{-1}Ly^{-1}$. It is context-free, as so
  is every left (and right) quotient of a context-free language by a
  regular language. Then we consider the three-letter alphabet
  $\{a,b,c\}$ and the homomorphism $\varphi : \{a,b,c\}^* \rightarrow
  A^*$ given by the images of letters: $\varphi(a)=v$, $\varphi(b)=w$,
  and $\varphi(c)=x$. We denote $K=\varphi^{-1} (u^{-1}Ly^{-1})$, the
  preimage in that morphism, which is context-free because the class
  of all context-free languages is closed under taking preimages under
  homomorphisms.

  Now, we apply Lemma~\ref{l:pumping-a-and-c-to-n-factorial} to the
  language $K$. Thus we have $n_0$ with the described property. This
  $n_0$ serves also for the statement of the proposition. Let $n$ be
  an arbitrary integer such that $n\ge n_0$ and $uv^{n!}wx^{n!}y\in
  L$. Then $a^{n!}bc^{n!} \in K$. By the property in
  Lemma~\ref{l:pumping-a-and-c-to-n-factorial}, there is $k<n_0\le n$
  such that, for every $p\in \mathbb N$, we have $a^{n!+pk}bc^{n!+pk}
  \in K$. It follows that $v^{n!+pk}wx^{n!+pk}\in u^{-1}Ly^{-1}$ and,
  consequently, $uv^{n!+pk}wx^{n!+pk}y\in L$. This gives the first
  statement in the proposition.
  
  In particular, since $k<n$, the number $k$ divides $n!$, and we may
  use $uv^{n!+pk}wx^{n!+pk}y\in L$ with $p=\frac{n!}{k}$. In this way,
  we obtain the required property $uv^{2(n!)}wx^{2(n!)}y\in L$.
\end{proof}

As an application of Proposition~\ref{p:pumping-a-to-n-factorial-b-c},
we can show that the language $L=\{a^{n!}bc^{n!}: n\in\mathbb{N}\}$ is
not an intersection of finitely many context-free languages. Indeed,
otherwise, if we take $A=\{a,b,c\}$, $x=a$, $y=b$, and $z=c$, from the
proposition it would follow that
Condition~(\ref{eq:pumping-a-to-n-factorial-b-c}) holds. However,
$a^{n!}bc^{n!}\in L$ holds for every $n$ while
$a^{2(n!)}bc^{2(n!)}\notin L$ for every $n\ge2$. Hence $L$ cannot
belong to~$\mathbf{ICF}(A)$. It amounts to an exercise to show that
$L$ is context-sensitive by showing that it is recognized by a linear
bounded automaton.

As the preceding example may seem somewhat cooked up, we proceed to
present another example which may be considered much more natural,
namely that of the language $\{a^mb^{mn}: m,n\ge1\}$, which can be
shown to be recognized by a two-way pushdown automaton
\cite[Exercise~7.13]{Shallit:2009}, whence recognized in linear time
\cite{Cook:1972} and context-sensitive
\cite[Theorem~42]{Gray&Harrison&Ibarra:1967}. We apply again the
general separation criterion given by
Theorem~\ref{t:inequality-definability-sequences} by considering the
words $s_n=a^{n!}b^{(n!)^2}$ and $t_n=a^{n!}b^{(n!)^2+(n-1)!}$. Then, it
suffices to establish the following result.

\begin{Lemma}
  \label{l:pump-a-fac-b-fac-sq}
  Let $A=\{a,b\}$. Then the following condition holds for every
  language $L\in\mathbf{CF}(A)$:
  \begin{equation}
    \label{eq:pump-a-fac-b-fac-sq}
    \exists n_0\,\forall n\ge n_0\
    (a^{n!}b^{(n!)^2} \in L \implies a^{n!}b^{(n!)^2+(n-1)!}\in L).
  \end{equation}
\end{Lemma}

\begin{proof}
  The idea of the proof is the same as for
  Lemma~\ref{l:pumping-a-and-c-to-n-factorial}. So, we start by
  considering the homomorphism $\varphi:A^*\to \mathbb{N}^2$ defined
  by $\varphi(w)=(|w|_a,|w|_b)$. By Parikh's theorem, the set
  $\varphi(L\cap a^*b^*)$ is semi-linear. Let $m$ be the maximum
  coordinate of all the vectors that appear in a concrete semi-linear
  expression for $\varphi(L\cap a^*b^*)$ and let $n_0=m+1$. We claim
  that Condition~(\ref{eq:pump-a-fac-b-fac-sq}) holds for this choice
  of $n_0$.

  Let $n\ge n_0$ and assume that $a^{n!}b^{(n!)^2} \in L$. Then there
  is a linear component
  \begin{displaymath}
    S=\{(p,q)+\sum_{i=1}^\ell t_i(p_i,q_i):t_i\in\mathbb{N}\}
  \end{displaymath}
  among those in the chosen semi-linear expression for $\varphi(L\cap
  a^*b^*)$, where we may as well assume that $p_i=q_i=0$ never holds,
  such that $(n!,(n!)^2)\in S$, say the coefficients
  $t_1,\ldots,t_\ell$ witness that fact. If $p_i\ne0$ for every $i$
  then, from 
  \begin{displaymath}
    \sum_{i=1}^\ell t_ip_i
    \le p+\sum_{i=1}^\ell t_ip_i
    =n!
  \end{displaymath}
  we conclude that
  \begin{align*}
    (n!)^2
    &=q+\sum_{i=1}^\ell t_iq_i
    \le m+\sum_{i=1}^\ell t_im \\
    &\le m+n!\,m
    <n!\,(m+1)
    =n!\,n_0
    \le(n!)^2.
  \end{align*}
  Hence, we must have $p_i=0$ for some~$i$, so that $q_i\ne0$ and,
  letting $t_i'=t_i+ \frac{(n-1)!}{q_i}$ and $t'_j=t_j$ for $j\ne i$
  we get $(n!,(n!)^2+(n-1)!)=(p,q)+\sum_{k=1}^\ell t'_k(p_k,q_k)$.
  This shows that $a^{n!}b^{(n!)^2+(n-1)!}\in L$, as claimed.
\end{proof}

Note that, more generally, the above argument shows that no language
contained in $\{a^mb^{mn}:m,n\ge1\}$ and containing infinitely many
words of the form $a^{n!}b^{e_n}$ with $e_n\ge n!\,n$ belongs
to~$\mathbf{ICF}(A)$ for any finite alphabet containing the letters
$a$ and $b$. This is the case of the language $\{a^nb^{n^2}:n\ge1\}$
considered in \cite[Exercise~7.14]{Shallit:2009}.

A similar argument shows that there is no set contained in the language
$\{a^mb^nc^{mn}:m,n\ge1\}$ (considered in~\cite{enwiki:1226028814})
and containing infinitely many words of the form
$a^{n!}b^{n!}c^{(n!)^2}$ that belongs to $\mathbf{ICF}(A)$.

\bibliographystyle{amsxport}
\bibliography{sgpabb,ref-sgps}

\end{document}